\documentclass[sigconf]{aamas} 

\usepackage{balance} 

\usepackage{amsmath,amsfonts,bm}

\def\eqref#1{equation~\ref{#1}}
\def\Eqref#1{Equation~\ref{#1}}

\def\1{\bm{1}}

\DeclareMathAlphabet{\mathsfit}{\encodingdefault}{\sfdefault}{m}{sl}
\SetMathAlphabet{\mathsfit}{bold}{\encodingdefault}{\sfdefault}{bx}{n}

\DeclareMathOperator*{\argmin}{arg\,min}

\usepackage{amsthm}

\makeatletter
\newtheorem*{rep@theorem}{\rep@title}
\newcommand{\newreptheorem}[2]{%
\newenvironment{rep#1}[1]{%
 \def\rep@title{#2 \ref{##1}}%
 \begin{rep@theorem}}%
 {\end{rep@theorem}}}
\makeatother

\newtheorem{theorem}{Theorem}
\newreptheorem{theorem}{Theorem}
\newtheorem{proposition}{Proposition}[section]

\newtheorem{lemma}{Lemma}
\newreptheorem{lemma}{Lemma}

\newtheorem{definition}{Definition}
\newreptheorem{definition}{Definition}

\newcommand\myeq{\mathrel{\overset{\makebox[0pt]{\mbox{\normalfont\tiny\sffamily def}}}{=}}}

\usepackage[utf8]{inputenc} 
\usepackage[T1]{fontenc}    
\usepackage{hyperref}       
\usepackage{url}            
\usepackage{booktabs}       
\usepackage{amsfonts}       
\usepackage{nicefrac}       
\usepackage{microtype}      

\usepackage{hyperref}

\usepackage{subcaption}
\usepackage{epigraph}

\usepackage{algorithm,algorithmic}

\newcounter{dbaCounter}

\newcounter{imgCounter}

\usepackage{amsmath,amsfonts,amsthm}
\usepackage{stmaryrd}  
\usepackage{dsfont}
\DeclareMathOperator*{\trace}{Tr}
\usepackage[normalem]{ulem}

\usepackage{xspace}
\usepackage{mathtools}
\DeclarePairedDelimiterX{\infdivx}[2]{(}{)}{%
  #1\;\delimsize|\delimsize|\;#2%
}
\newcommand{\kld}[2]{\ensuremath{\mathcal{D}_{KL}\infdivx{#1}{#2}}\xspace}

\usepackage{cleveref}

\usepackage{enumitem}
\setitemize{noitemsep,topsep=0pt,parsep=0pt,partopsep=0pt}

\usepackage{url}

\usepackage[normalem]{ulem}

\usepackage[toc,page]{appendix}
\addtocontents{toc}{\protect\setcounter{tocdepth}{-1}}

\setcopyright{ifaamas}
\acmConference[AAMAS '22]{Proc.\@ of the 21st International Conference
on Autonomous Agents and Multiagent Systems (AAMAS 2022)}{May 9--13, 2022}
{Online}{P.~Faliszewski, V.~Mascardi, C.~Pelachaud,
M.E.~Taylor (eds.)}
\copyrightyear{2022}
\acmYear{2022}
\acmDOI{}
\acmPrice{}
\acmISBN{}

\acmSubmissionID{285}

\title[D3C]{D3C: Reducing the Price of Anarchy in Multi-Agent Learning}

\author{Ian Gemp}
\affiliation{
  \institution{DeepMind}
  \city{London}
  \country{United Kingdom}}
\email{imgemp@deepmind.com}

\author{Kevin R. McKee}
\affiliation{
  \institution{DeepMind}
  \city{London}
  \country{United Kingdom}}
\email{kevinrmckee@deepmind.com}

\author{Richard Everett}
\affiliation{
  \institution{DeepMind}
  \city{London}
  \country{United Kingdom}}
\email{reverett@deepmind.com}

\author{Edgar Du{\'e}{\~n}ez-Guzm{\'a}n}
\affiliation{
  \institution{DeepMind}
  \city{London}
  \country{United Kingdom}}
\email{duenez@deepmind.com}

\author{Yoram Bachrach}
\affiliation{
  \institution{DeepMind}
  \city{London}
  \country{United Kingdom}}
\email{yorambac@deepmind.com}

\author{David Balduzzi}
\affiliation{
  \institution{XTX Markets}
  \city{London}
  \country{United Kingdom}}
\email{dbalduzzi@gmail.com}

\author{Andrea Tacchetti}
\affiliation{
  \institution{DeepMind}
  \city{London}
  \country{United Kingdom}}
\email{atacchet@deepmind.com}

\begin{abstract}
In multiagent systems, the complex interaction of fixed incentives can lead agents to outcomes that are poor (\emph{inefficient}) not only for the group, but also for each individual. Price of anarchy is a technical, game-theoretic definition that quantifies the inefficiency arising in these scenarios\textemdash it compares the welfare that can be achieved through perfect coordination against that achieved by self-interested agents at a Nash equilibrium. We derive a differentiable, upper bound on a price of anarchy that agents can cheaply estimate during learning. Equipped with this estimator, agents can adjust their incentives in a way that improves the efficiency incurred at a Nash equilibrium. Agents do so by learning to mix their reward (equiv. negative loss) with that of other agents by following the gradient of our derived upper bound. We refer to this approach as D3C. In the case where agent incentives are differentiable, D3C resembles the celebrated Win-Stay, Lose-Shift strategy from behavioral game theory, thereby establishing a connection between the global goal of maximum welfare and an established agent-centric learning rule. In the non-differentiable setting, as is common in multiagent reinforcement learning, we show the upper bound can be reduced via evolutionary strategies, until a compromise is reached in a distributed fashion. We demonstrate that D3C improves outcomes for each agent and the group as a whole on several social dilemmas including a traffic network exhibiting Braess’s paradox, a prisoner’s dilemma, and several multiagent domains.
\end{abstract}


\newcommand{\BibTeX}{\rm B\kern-.05em{\sc i\kern-.025em b}\kern-.08em\TeX}

\begin{document}


\pagestyle{fancy}
\fancyhead{}


\maketitle 


\section{Introduction}
\label{intro}





We consider a setting consisting of many interacting artificially intelligent agents, each with specific individual incentives. It is well known that the interactions between individual agent goals can lead to inefficiencies at the group level, for example, in environments exhibiting social dilemmas \citep{braess1968paradoxon, hardin1968tragedy, leibo2017multi}. In order to resolve these fundamental inefficiencies, agents must reach a compromise.


Any arbitration mechanism with a central coordinator\footnote{For example, the Vickrey-Clarke-Groves (VCG) mechanism~\citep{clarke1971multipart}.} faces challenges when scaling to large populations. The coordinator's task becomes intractable as it must both query preferences from a larger population and make decisions accounting for the exponential growth of agent interactions. If agents are permitted to modify their incentives over time, the coordinator must collect all this information again, exacerbating the computational burden. In addition, a central coordinator represents a single point of failure for the system whereas successful multiagent systems identified in nature (e.g., market economies, ant colonies, etc.) are often robust to node failures~\citep{edelman2001degeneracy}. Therefore, we focus on decentralized approaches.

\textbf{Design Criteria}:
The celebrated Myerson-Satterthwaite theorem~\citep{arrow2012social,satterthwaite1975strategy,green1977characterization,myerson1983efficient} states that no mechanism can simultaneously achieve optimal \uline{efficiency} (welfare-maximizing behavior), \uline{budget-} \uline{balance} (no taxing agents, burning side-payments, or hallucinating rewards), appeal to \uline{rational individuals} (individuals want to opt-in to the mechanism), and be \uline{incentive compatible} (resulting behavior is a Nash equilibrium).
While this impossibility result precludes a mechanism that satisfies the above criteria perfectly, it says nothing about a mechanism that satisfies them \emph{approximately}, which is our aim here.
In addition, the mechanism should be decentralized, extensible to large populations, and adapt to learning agents with evolving incentives in possibly non-stationary environments.

\textbf{Design}:
We formulate compromise as agents mixing their incentives (rewards or losses) with others. In other words, an agent may become incentivized to minimize a mixture of their loss and other agents' losses. We design a decentralized meta-algorithm that allows agents to search over the space of these possible mixtures.

We model the problem of \uline{efficiency} using \textit{price of anarchy}. The price of anarchy, $\rho \in [1,\infty)$, is a measure of inefficiency from algorithmic game theory with lower values indicating more efficient games~\citep{nisan2007algorithmic}. Forcing agents to minimize a group (average) loss with a single local minimum results in a ``game'' with $\rho = 1$. Note that any optimal group loss solution is also Pareto-efficient. Computing the price of anarchy of a game is intractable in general. Instead, we derive a differentiable upper bound on the price of anarchy that agents can optimize incrementally over time. Differentiability of the bound makes it easy to pair the proposed mechanism with, for example, deep learning agents that optimize via gradient descent~\citep{lerer2017maintaining,openai2019dota}. \uline{Budget balance} is achieved exactly by placing constraints on the allowable mixtures of losses. We appeal to \uline{individual rationality} in three ways. One, we initialize all agents to optimize only their own losses. Two, we include penalties for agents that deviate from this state and mix their losses with others. Three, we show empirically on several domains that opting into the proposed mechanism results in better individual outcomes. We also provide specific, albeit narrow, conditions under which agents may achieve a Nash equilibrium, i.e. the mechanism is \uline{incentive compatible}, and demonstrate the agents achieving a Nash equilibrium under our proposed mechanism in a traffic network problem. Note that budget-balance is the only property we guarantee is satisfied in absolute terms. All other properties are appealed to either indirectly via design choices (e.g., minimizing $\rho$) or post-hoc analysis.

{\bf Our Contribution:} We propose a differentiable, local estimator of game inefficiency, as measured by price of anarchy. We then present two instantiations of a single decentralized meta-algorithm, one $1$st order (gradient-feedback) and one $0$th order (bandit-feedback), that reduce this inefficiency. This meta-algorithm is general and can be applied to any group of individual agent learning algorithms. In contrast to the centralized training, decentralized execution framework popular in multiagent reinforcement learning (MARL), we demonstrate the success of our meta-algorithm in a more challenging online setting (decentralized training, decentralized execution) on a range of games and MARL domains.
%


This paper focuses on how to enable a group of agents to respond to an unknown environment and minimize overall inefficiency. Agents with distinct losses may find their incentives well aligned to the given task, however, they may instead encounter a \emph{social dilemma} (\S\ref{experiments}). We also show that our approach leads to sensible behavior in scenarios where agents may need to sacrifice team reward to \emph{save an individual} (Appx.~\ref{ineq_aver}) or need to form parties and \emph{vote} on a new team direction (Appx.~\ref{election}). Ideally, one meta-algorithm would allow a multiagent system to perform sufficiently well in all these scenarios.
The approach we propose, D3C (\S\ref{d3c}), represents a holistic effort to design such a meta-algorithm.\footnote{D3C is agnostic to any action or strategy semantics. We are interested in rich environments where high level actions with semantics such as ``cooperation'' and ``defection'' are not easily extracted or do not exist.}

\section{Dynamically Changing the Game}
\label{d3c}
In our approach, agents may consider slight re-definitions of their original losses, thereby changing the definition of the original game. Critically, this is done in a way that conserves the original sum of losses (budget-balanced) so that the original group loss can still be measured.
In this section, we derive our approach to minimizing the price of anarchy in several steps. First we formulate minimizing the price of anarchy via compromise as an optimization problem. Second we specifically consider compromise as the linear mixing of agent incentives. Next, we define a \emph{local} price of anarchy and derive an upper bound that agents can differentiate. Then, we decompose this bound into a set of differentiable objectives, one for each agent. Finally, we develop a gradient estimator to minimize the agent objectives in settings with bandit feedback (e.g., RL) that enables scalable decentralization.

\subsection{Notation and Transformed Losses}
Let agent $i$'s loss be $f_i(\boldsymbol{x}): \boldsymbol{x} \in \mathcal{X} \rightarrow \mathbb{R}$ where $\boldsymbol{x}$ is the joint strategy of all agents. Let $f_i^A(\boldsymbol{x})$ denote agent $i$'s transformed loss which mixes losses among agents. Let $\boldsymbol{f}(\boldsymbol{x}) = [f_1(\boldsymbol{x}), \ldots, f_n(\boldsymbol{x})]^\top$ and $\boldsymbol{f}^A(\boldsymbol{x}) = [f^A_1(\boldsymbol{x}), \ldots, f^A_n(\boldsymbol{x})]^\top$ where $n \in \mathbb{Z}$ denotes the number of agents. In general, we require $f_i^A(\boldsymbol{x})>0$ and $\sum_i f_i^A(\boldsymbol{x}) = \sum_i f_i(\boldsymbol{x})$ so that total loss is conserved\footnote{The strict definition of price of anarchy assumes positive losses. This is relaxed in~\S\ref{decentralized} to allow for losses in $\mathbb{R}$.}. Under these constraints, the agents will simply explore the space of possible non-negative group loss decompositions. We consider transformations of the form $\boldsymbol{f}^A(\boldsymbol{x}) = A^\top \boldsymbol{f}(\boldsymbol{x})$ (note the tranpose) where each agent $i$ controls row $i$ of $A$ with each row constrained to the simplex, i.e. $A_i \in \Delta^{n-1}$. For example, agent $1$'s loss is mixed according to the first \textbf{column} of $A$ which may not sum to $1$, and not the first row, which it controls:
%
\begin{align}
    f_1^A(\boldsymbol{x}) &= \langle \overbrace{[0.9, 0.3, 0.5]}^{[A_{11}, A_{21}, A_{31}]}, [f_1(\boldsymbol{x}), f_2(\boldsymbol{x}), f_3(\boldsymbol{x})] \rangle. \label{mixing_example}
\end{align}
Lastly, $[a;b]=[a^\top,b^\top]^\top$ signifies row stacking of vectors, and $\mathcal{X}^*$ denotes the set of Nash equilibria.

\subsection{Price of Anarchy}
\citet{nisan2007algorithmic} define price of anarchy as the worst value of an equilibrium divided by the best value in the game. Here, value means sum of player losses, best means lowest, and Nash is the chosen equilibrium concept. It is well known that Nash can be arbitrarily bad from both an individual agent and group perspective; Appx.~\ref{bad_nash} presents a simple example and demonstrates how opponent shaping~\citep{foerster2018learning,letcher2018stable} is not a balm for these issues. With the above notation, the price of anarchy is defined as
\begin{align}
    \rho_{\mathcal{X}}(\boldsymbol{f}^A) &\myeq \frac{\max_{\mathcal{X}^*} \sum_i f^A_i(\boldsymbol{x}^*)}{\min_{\mathcal{X}} \sum_i f^A_i(\boldsymbol{x})} \ge 1. \label{global_poa_util_def}
\end{align}

Note that computing the price of anarchy precisely requires solving for both the optimal welfare and the worst case Nash equilibrium. We explain how we circumvent this issue with a local approximation in \S\ref{local_poa_subsec}.

\subsection{Compromise as an Optimization Problem}
Given a game, we want to minimize the price of anarchy by perturbing the original agent losses:
\begin{align}
    \min_{\substack{\boldsymbol{f}'=\psi_{A}(\boldsymbol{f}) \\ \boldsymbol{1}^\top \boldsymbol{f}' = \boldsymbol{1}^\top \boldsymbol{f} }} \rho_{\mathcal{X}}(\boldsymbol{f}') + \nu \mathcal{D}(\boldsymbol{f}, \boldsymbol{f}') \label{rho_minimization}
\end{align}
where $\boldsymbol{f}$ and $\boldsymbol{f}'=\psi_{A}(\boldsymbol{f})$ denote the vectors of original and perturbed losses respectively, $\psi_{A}: \mathbb{R}^n \rightarrow \mathbb{R}^n$ is parameterized by weights $A$, $\nu$ is a regularization hyperparameter, and $\mathcal{D}$ penalizes deviation of the perturbed losses from the originals or represents constraints through an indicator function. To ensure minimizing the price of anarchy of the perturbed game improves on the original, we incorporate the constraint that the sum of perturbed losses equals the sum of original losses, $\boldsymbol{1}^\top \boldsymbol{f}' = \boldsymbol{1}^\top \boldsymbol{f}$. We refer to this approach as $\rho$-minimization.

Our agents reconstruct their losses using the losses of all other agents as a basis. For simplicity, we consider linear transformations of their loss functions, although the theoretical bounds hereafter are independent of this simplification. We also restrict ourselves to convex combinations so that agents do not learn incentives that are directly adverse to other agents. The problem can now be reformulated. Let $\psi_A(\boldsymbol{f}) = A^\top \boldsymbol{f}$ and $\mathcal{D}(\boldsymbol{f},\boldsymbol{f}') = \sum_i \kld{\boldsymbol{e}_i}{A_i}$
where $A \in \mathbb{R}^{n \times n}$ is a right stochastic matrix (rows are non-negative and sum to $1$), $\boldsymbol{e}_i \in \mathbb{R}^n$ is a unit vector with a $1$ at index $i$, and $\mathcal{D}_{KL}$ denotes the Kullback-Liebler divergence.

\subsection{A Local Price of Anarchy}
\label{local_poa_subsec}
The price of anarchy, $\rho \ge 1$, is defined over the joint strategy space of all players. Computing it is intractable for general games.
However, many agents learn via gradient-based training, and so only observe the portion of the strategy space explored by their learning trajectory. Hence, we imbue our agents with the ability to locally estimate the price of anarchy along this trajectory.
\begin{definition}
[\emph{Local} Price of Anarchy] Define
\begin{align}
    \rho_{\boldsymbol{x}}(\boldsymbol{f}^A, \Delta t) &= \frac{\max_{\mathcal{X}^*_{\tau}} \sum_i f^A_i(\boldsymbol{x}^*)}{\min_{\tau\in[0,\Delta t]} \sum_i f^A_i(\boldsymbol{x}-\tau F(\boldsymbol{x}))} \ge 1 \label{poa_util_def}
\end{align}
where $F(\boldsymbol{x}) = [\nabla_{x_1}f^A_1(\boldsymbol{x}); \ldots; \nabla_{x_n}f^A_n(\boldsymbol{x})]$, $\Delta t$ is a small step size, $f_i^A$ is assumed positive $\forall \, i$, and $\mathcal{X}›_{\tau}$ denotes the set of equilibria of the game when constrained to the line.
\end{definition}
\begin{figure}[ht]
    \centering
    \includegraphics[scale=0.2]{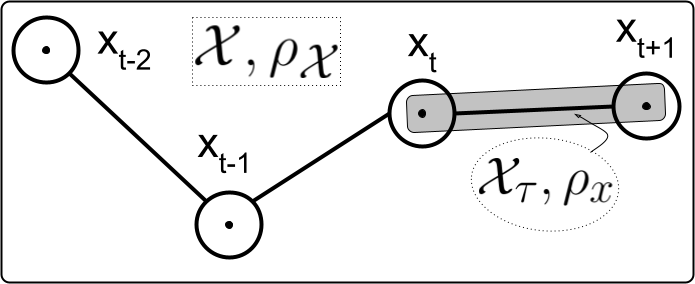}
    \caption{Agents estimate the price of anarchy assuming the joint strategy space, $\mathcal{X}$, of the game is restricted to a local linear region, $\mathcal{X_{\tau}}$, extending from the currently learned joint strategy, $x_t$, to the next, $x_{t+1}$. $\rho_{\mathcal{X}}$ and $\rho_x$ denote the global and local price of anarchy.}
    \label{fig:local_poa}
\end{figure}

To obtain bounds, we leverage theoretical results on \emph{smooth games}, summarized as a class of games where ``the externality imposed on any one player by the others is bounded''~\citep{roughgarden2015intrinsic}. We assume a Lipschitz property on all $f^A_i(\boldsymbol{x})$ (details in Theorem~\ref{poa_theorem}), which allows us to appeal to this class of games. The bound in \eqref{poa_theorem_bnd} is tight for some games. Proofs can be found in Appx.~\ref{bound_proofs}.

For convenience, we repeat the core definition and lemma put forth by~\citet{roughgarden2015intrinsic} here.

\begin{definition}[Smooth Game]
\label{def:smooth_game}
    A game is $(\lambda, \mu)$-smooth~\citep{roughgarden2015intrinsic} if:
    \begin{align}
        \sum_{i=1}^n f^A_i(x_i,x'_{-i}) &\le \lambda \sum_{i=1}^n f^A_i(x_i,x_{-i}) + \mu \sum_{i=1}^n f^A_i(x'_i,x'_{-i}) \label{eq:smooth}
    \end{align}
    for all $\boldsymbol{x}, \boldsymbol{x}' \in \mathcal{X}$ where $\lambda > 0$, $\mu < 1$. $x_{-i}$ denotes all player $j \ne i$ strategies and $\sum_i f_i^A(\boldsymbol{x})$ is assumed to be non-negative for any $\boldsymbol{x} \in \mathcal{X}$.
\end{definition}
The last condition is needed for the price of anarchy, a ratio of welfares, to be meaningful as a positive measure of inefficiency.

\begin{lemma}[Smooth Games Imply a Bound on Price of Anarchy]
\label{lemma:smooth_to_poa}
The price of anarchy is bounded above by a ratio of the coefficients that satisfy the smooth game definition~\citep{roughgarden2015intrinsic}:
\begin{align}
    1 \le \rho_{\mathcal{X}}(\boldsymbol{f}^A) &\le \inf_{\lambda>0, \mu<1} \Big[ \frac{\lambda}{1-\mu} \Big]. \label{eq:poa_opt}
\end{align}
\end{lemma}

\begin{theorem}[Local \emph{Utilitarian} Price of Anarchy]
\label{poa_theorem}
Assuming each agent's loss is positive and its loss gradient is Lipschitz, there exists a learning rate $\Delta t > 0$ sufficiently small such that, to $\mathcal{O}(\Delta t^2)$, the local \textbf{utilitarian} price of anarchy of the game, $\rho_{\boldsymbol{x}}(\boldsymbol{f}^A, \Delta t)$, is upper bounded by
\begin{align}
    &\max_i \{1 + \Delta t \, \emph{\texttt{ReLU}} \Big( \frac{d}{dt} \log(f_i^A(\boldsymbol{x})) + \frac{||\nabla_{x_i} f_i^A(\boldsymbol{x})||^2}{\bar{\mu} f_i^A(\boldsymbol{x})} \Big) \} \label{poa_theorem_bnd}
\end{align}
where $i$ indexes each agent, $\bar{\mu} \in \mathbb{R}_{\ge 0}$ is a user-defined upper bound on the true $\mu$, $\texttt{ReLU}(z) \myeq \max(z, 0)$, and Lipschitz implies there exists a $\beta_i$ such that $||\nabla_{x_i} f_i^A(\boldsymbol{x}) - \nabla_{y_i} f_i^A(\boldsymbol{y}) || \le \beta_i ||\boldsymbol{x}-\boldsymbol{y}|| \,\, \forall \boldsymbol{x}, \boldsymbol{y}, A$.
\end{theorem}

\paragraph{Proof Sketch:}
For a small enough region (\textcolor{gray}{grayed} in Figure~\ref{fig:local_poa}), we can approximate each agent's loss function with its Taylor series expansion. By rewriting all losses in the smoothness constraint (\eqref{eq:smooth}) in terms of expansions about $\boldsymbol{x}$ or $\boldsymbol{x'}$, i.e., quantities we can measure before and after a joint gradient step, we can proceed to define the smoothness constraint with $\mu$ and $\lambda$ in terms of measurable quantities. The smoothness constraint is formulated as a sum over the $n$ agents, but we can decompose this constraint into $n$ individual constraints with their own $\mu_i$'s and $\lambda_i$'s. If each agent can ensure local individual smoothness, which is possible for a small enough region, we show this is sufficient to satisfy the original local smoothness condition with $\mu=\max_i \{\mu_i\}$ and $\lambda=\max_i \{\lambda_i\}$. Each agent can further estimate their own individual price of anarchy, $\rho_i$, via \eqref{eq:poa_opt} which reduces to a tractable two dimensional constrained optimization problem with a closed form solution. We further show that we can upper bound the local price of anarchy for the group (\eqref{poa_util_def}) with the max of these individual estimates. Finally, using another expansion along with the log-trick famous from the policy gradient theorem, we recover the final result presented in Theorem~\ref{poa_theorem} below. The Lipschitz assumption exists simply to ensure the series approximations are sufficiently accurate for a small enough region. The full proof is in Appx.~\ref{bound_proofs}.

Recall that this work focuses on price of anarchy defined using total loss as the value of the game. This is a \emph{utilitarian} objective. We also derive an upper bound on the local \emph{egalitarian} price of anarchy where value is defined as the max loss over all agents (replace $\sum_i$ with $\max_i$ in~\eqref{poa_util_def}; see~Appx.~\ref{egalitarian}), possibly of independent interest.

\begin{theorem}
\label{poa_theorem_egalitarian}
Given $n$ positive losses, $f_i^A(\boldsymbol{x})$, $i\in \{1,\ldots,n\}$, with $\beta_i$-Lipschitz gradients there exists a $\Delta t > 0$ sufficiently small such that, to $\mathcal{O}(\Delta t^2)$, the local \textbf{egalitarian} price of anarchy of the game is upper bounded by
\begin{equation}
    \rho_e \le 1 + \Delta t \, \texttt{ReLU} \Big( \frac{d}{dt} \log(\max_i\{f_i^A(\boldsymbol{x})\}) + \frac{\sum_{i=1}^n ||\nabla_{x_i} f_i^A(\boldsymbol{x})||^2}{\bar{\mu} \max_i \{ f_i^A(\boldsymbol{x}) \}} \Big).
\end{equation}
\end{theorem}



\subsection{Decentralized Learning of the Loss Mixture Matrix $A$}
\label{decentralized}
Minimizing~\eqref{rho_minimization} w.r.t. $A$ can become intractable if $n$ is large. Moreover, if solving for $A$ at each step is the responsibility of a central authority, the system is vulnerable to this authority failing. A distributed solution is therefore appealing, and the local price of anarchy bound admits a natural relaxation that decomposes over agents ($\max_i z_i \le \sum_i z_i$ for $z_i \ge 0$). \Eqref{rho_minimization} then factorizes as
\begin{align}
    \min_{A_i \in \Delta^{n-1}} \rho_i + \nu \kld{\boldsymbol{e}_i}{A_i} \label{lin_poa_dist}
\end{align}
where $\rho_i = 1 + \Delta t \, \texttt{ReLU} \Big( \frac{d}{dt} \log(f_i^A(\boldsymbol{x})) + \frac{||\nabla_{x_i} f_i^A(\boldsymbol{x})||^2}{f_i^A(\boldsymbol{x})\bar{\mu}} \Big)$.
%
%
Local price of anarchy is subdifferentiable w.r.t. each $A_i$ with gradient
\begin{align}
\nabla_{A_i} \rho_i \propto \nabla_{A_i} \texttt{ReLU} \Big( \frac{d}{dt} \log(f_i^A(\boldsymbol{x})) + \frac{||\nabla_{x_i} f_i^A(\boldsymbol{x})||^2}{f_i^A(\boldsymbol{x})\bar{\mu}} \Big).
\end{align}
The $\log$ appears due to price of anarchy being defined as the worst case Nash total loss \emph{divided} by the minimal total loss. We propose the following modified learning rule for a hypothetical price of anarchy which is defined as a \emph{difference} and accepts negative loss: $A_i \leftarrow A_i - \eta_A \tilde{\nabla}_{A_i} (\rho_i + \nu \mathcal{D}_{KL})$ where $\eta_A$ is a learning rate and
\begin{align}
    \textcolor{blue}{\tilde{\nabla}_{A_i} \rho_i} &= \textcolor{blue}{\nabla_{A_i} \texttt{ReLU} \Big( \frac{d}{dt} f_i^A(\boldsymbol{x}) + \epsilon \Big)}. \hspace{0.2cm} \text{[$\epsilon$ is a hyperparameter.]} \label{poa_additive}
\end{align}
The update direction in~(\ref{poa_additive}) is proportional to $\nabla_{A_i} \rho_i$ asymptotically for large $f_i^A$; see~Appx.~\ref{multiplicative_poa} for further discussion. Each agent $i$ updates $x_i$ and $A_i$ simultaneously using $\nabla_{x_i} f_i^A(\boldsymbol{x})$ and $\tilde{\nabla}_{A_i} (\rho_i + \nu \mathcal{D}_{KL})$.
%

\textbf{Improve-Stay, Suffer-Shift}\textemdash Win-Stay, Lose-Shift (WSLS)~\citep{robbins1952some} is a strategy shown to outperform Tit-for-Tat~\citep{rapoport1965prisoner} in an iterated prisoner's dilemma~\citep{nowak1993strategy,imhof2007tit}. It was also shown to be psychologically plausible~\citep{wang2014social} in research on human play. The D3C update direction, $\nabla_{A_i} \rho_i$, encodes the rule: if the loss is decreasing, maintain the mixing weights, otherwise, change them. We can interpret this rule as a generalization of WSLS to learning (derivatives) rather than outcomes (losses). Therefore, we have shown that a sensible, agent-centric learning rule (WSLS) can be derived from minimization of the global, game theoretic concept \emph{price of anarchy} by simply a) restricting agents' strategy spaces to be local to their learning trajectory, a form of bounded rationality, and b) having the agents consider improvements (derivatives) instead of direct outcomes. Furthermore, the fact that a lower price of anarchy entails a higher welfare at a Nash equilibrium means this style of WSLS is ultimately compatible with achieving high performance for the entire system.


Note that the trival solution of minimizing average group loss coincides with $A_{ij} = \frac{1}{n}$ for all $i,j$. If the agent strategies converge to a social optimum, this is a fixed point in the augmented strategy space $(\boldsymbol{x},A)$. This can be seen by noting that 1) convergence to an optimum implies $\nabla_{x_i} f_i^A(\boldsymbol{x}) = 0$ and 2) convergence alone implies $\frac{df_i}{dt} = 0$ for all agents so $\nabla A_i = 0$ by~\eqref{poa_additive} assuming $\epsilon=0$.

\subsection{Decentralized Learning \& Extending to Reinforcement Learning}
The time derivative of each agent's loss, $\frac{d}{dt} f_i^A(\boldsymbol{x})$, in \eqref{poa_additive} requires differentiating through potentially all other agent loss functions, which precludes scaling to large populations. In addition, this derivative is not always available as a differentiable function. In order to estimate $\tilde{\nabla}_{A_i} \rho_i$ when only scalar estimates of $\rho_i$ are available as in, e.g., multiagent reinforcement learning (MARL), each agent perturbs their loss mixture and commits to this perturbation for a random number of training steps. If the loss increases over the trial, the agent updates their mixture in a direction \emph{opposite} the perturbation. Otherwise, no update is performed.

This is formally accomplished with approximate one-shot gradient estimates~\citep{shalev2012online} or \emph{evolutionary strategies}~\citep{rechenberg1978evolutionsstrategien}. A one-shot gradient of $\rho_i(A_i)$ is estimated by first perturbing $A_i$ with entropic mirror ascent~\citep{beck2003mirror} as $\tilde{A}_i = \texttt{softmax}(\log(A_i) + \delta \tilde{\boldsymbol{a}}_i)$ where $\delta > 0$ and $\tilde{\boldsymbol{a}_i} \sim U_{sp}(n)$ is drawn uniformly from the unit sphere in $\mathbb{R}^n$. The perturbed weights are then evaluated $\tilde{\rho}_i = \rho_i(\tilde{A}_i)$. Finally, an unbiased gradient is given by $\frac{n}{\delta} \tilde{\rho}_i \tilde{\boldsymbol{a}_i}$. In practice, we cannot evaluate in one shot the $\frac{d}{dt} f_i^A(\boldsymbol{x})$ term that appears in the definition of $\rho_i$. Instead, Algorithm~\ref{alg_rl_top} uses finite differences and we assume the evaluation remains accurate enough across training steps.
\begin{algorithm}[H]
\begin{algorithmic}
    \STATE Input: $\eta_A$, $\delta$, $\nu$, $\tau_{\min}$, $\tau_{\max}$, $A_i^0$, $\epsilon$, $l$, $h$, $\mathbb{L}$, iterations $T$
    \STATE $A_i \leftarrow A_i^0$ \{\textcolor{orange}{Initialize Mixing Weights}\}
    \STATE $G = 0$ \{\textcolor{orange}{Initialize Mean Return of Trial}\}
    \STATE \{\textcolor{orange}{Draw Initial Random Mixing Trial}\}
    \STATE $\tilde{A}_i, \tilde{\boldsymbol{a}}_i, \tau, t_b, G_b = \texttt{trial}(\delta, \tau_{\min}, \tau_{\max}, A_i, 0, G)$
    \FOR {$t = 1 : T$}
        \STATE $g = \mathbb{L}_i(\tilde{A}_j\,\, \forall \,\,j)$ \{\textcolor{orange}{Update Policy With Mixed Rewards}\}
        \STATE $\Delta t_b = t - t_b$ \{\textcolor{orange}{Elapsed Trial Steps}\}
        \STATE $G = (G (\Delta t_b - 1) + g)/\Delta t_b$ \{\textcolor{orange}{Update Mean Return}\}
        \IF {$\Delta t_b == \tau$ \{\textcolor{orange}{Trial Complete}\}}
            \STATE $\textcolor{blue}{\tilde{\rho}_i = \texttt{ReLU}(\frac{G_b - G}{\tau} + \epsilon)}$ \{\textcolor{orange}{Approximate $\rho$}\}
            \STATE $\textcolor{blue}{\nabla_{A_i} = \tilde{\rho}_i \tilde{\boldsymbol{a}}_i - \nu \boldsymbol{e}_i \varoslash A_i}$ \{\textcolor{orange}{Estimate Gradient} \textcolor{blue}{\textemdash(\ref{poa_additive})}\}
            \STATE $A_i = \texttt{softmax} \,_{l}\lfloor \log(A_i)$$ - $$\eta_{A} \nabla_{A_i} \rceil^{h}$ \{\textcolor{orange}{Update}\}
            \STATE \{\textcolor{orange}{Draw New Random Mixing Trial}\}
            \STATE $\tilde{A}_i, \tilde{\boldsymbol{a}}_i, \tau, t_b, G_b = \texttt{trial}(\delta, \tau_{\min}, \tau_{\max}, A_i, t, G)$
        \ENDIF
    \ENDFOR
\end{algorithmic}
\caption{D3C Update for RL Agent $i$}
\label{alg_rl_top}
\end{algorithm}
\vspace{-2em}
\begin{algorithm}[H]
\begin{algorithmic}
    \STATE Input: $\delta$, $\tau_{\min}$, $\tau_{\max}$, $A_i$, $t$, $G$
    \STATE $\tilde{\boldsymbol{a}}_i \sim U_{sp}(n)$ \{\textcolor{orange}{Sample Perturbation Direction}\}
    \STATE $\tilde{A}_i = \texttt{softmax}(\log(A_i) + \delta \tilde{\boldsymbol{a}}_i)$ \{\textcolor{orange}{Perturb Mixture}\}
    \STATE $\tau \sim \texttt{Uniform}\{\tau_{\min}, \tau_{\max}\}$ \{\textcolor{orange}{Draw Random Trial Length}\}
    \STATE Output: $\tilde{A}_i, \tilde{\boldsymbol{a}}_i, \tau, t, G$
\end{algorithmic}
\caption{\texttt{trial}\textemdash helper function}
\label{alg_rl_pert}
\end{algorithm}
Algorithm~\ref{alg_rl_top} requires several arguments: $\eta_A$ is a global learning rate for each $A_i$, $\delta$ is a perturbation scalar for the one-shot gradient estimate, $\tau_{\min}$ and $\tau_{max}$ specify the lower and upper bounds for the duration of the mixing trial for estimating a finite difference of $\frac{d}{dt} f_i^A(\boldsymbol{x}) \approx -(G - G_b)/\tau$, $l$ and $h$ specify lower and upper bounds for clipping $A$ in logit space ($_{l}\lfloor \cdot \rceil^{h}$), and $\mathbb{L}_i$ (Algorithm~\ref{alg_learner}) represents any generic reinforcement learning algorithm augmented to take $A$ as input (in order to mix rewards) and outputs \emph{discounted return}. $\varoslash$ indicates elementwise division.
\begin{algorithm}[H]
\begin{algorithmic}
    \STATE Input: $\tilde{A} = [\tilde{A}_1; \ldots; \tilde{A}_n]$
    \WHILE{episode not terminal}
        \STATE draw action from agent policy
        \STATE play action and observe reward $r_i$
        \STATE broadcast $r_i$ to all agents
        \STATE update policy with $\tilde{r}_i = \sum_j \tilde{A}_{ji} r_j$
    \ENDWHILE
    \STATE Output: return over episode $g$
\end{algorithmic}
\caption{$\mathbb{L}_i$\textemdash example learner}
\label{alg_learner}
\end{algorithm}

\subsection{Assessment}
We assess Algorithm~\ref{alg_rl_top} with respect to our original design criteria. As described, agents perform gradient descent on a decentralized and local upper bound on the price of anarchy. Recall that a minimal global price of anarchy ($\rho=1$) implies that even the worst case Nash equilibrium of the game is socially optimal; similarly, Algorithm~\ref{alg_rl_top} searches for a locally socially optimal equilibrium. By design, $A_i \in \Delta^{n-1}$ ensures the approach is budget-balancing. We justify the agents learning weight vectors $A_i$ by initializing them to attend primarily to their own losses as in the original game. If they can minimize their original loss, then they never shift attention according to \eqref{poa_additive} because $\frac{df_i}{dt} \le 0$ for all $t$. They only shift $A_i$ if their loss increases. We also include a KL term to encourage the weights to return to their initial values. In addition, in our experiments with symmetric games, learning $A$ helps the agents' outcomes in the long run. We also consider experiments in Appx.~\ref{d3c_robust_in_pd} where only a subset of agents opt into the mechanism. If each agent's original loss is convex with diagonally dominant Hessian and the strategy space is unconstrained, the unique, globally stable fixed point of the game defined with mixed losses is a Nash (see Appx.~\ref{incentive_compat}). Exact gradients $\nabla_{A_i} \rho_i$ require each agent differentiates through all other agents' losses precluding a fully decentralized and scalable algorithm. We circumvent this issue with noisy oneshot gradients. All that is needed in terms of centralization is to share the mixed scalar rewards; this is cheap compared to sharing $x_i \in \mathbb{R}^{d \gg 1}$. The cost of communicating rewards may be mitigated by learning $A_i$ via sparse optimization or sampling but is outside the scope of this paper.

\subsection{Related Work}
\emph{Collective Intelligence} or COIN, surveyed in~\citep{wolpert1999introduction}, examines the problem of how to design reward functions for individual agents such that a decentralized multiagent system maximizes a global world utility function. \citet{wolpert1999introduction} describe several approaches taken by an array of diverse fields and motivate the creation of a collective intelligence as an important challenge. Follow-up works focus on aiding researchers in deriving static agent reward functions that are consistent with optimizing the desired world utility via, for instance, useful visualizations~\cite{agogino2005multi,agogino2008analyzing}. Unlike conventional COIN approaches, D3C learns agent reward functions dynamically through online interaction with the environment. On the other hand, like D3C, studies in COIN find that agents optimizing modified versions of their original reward functions not only achieve high global utility, but also perform better individually~\cite{tumer2013coordinating}. 

In recent MARL work, \citet{lupu2020gifting} augment the agents' action space with a ``gifting'' action where agents can send a $+1$ reward to another agent. They evaluate this approach on a variant of \emph{Harvest} we explore in Appx.~\ref{harvestpatch}. They look at three different reward budget settings; ours is most similar to their zero-sum setting in which gifts are \emph{budget-balanced} by matching $-1$ penalties. In contrast to~\citep{lupu2020gifting}, we consider a continuum of ``gifting'' amounts automatically grounded in the scale of the original rewards via mixing on the simplex. 

Similarly,~\citet{hostallero2020inducing} introduce PED-DQN where agents gift their peers by a reciprocal amount proportional to the positive externality they perceive (as measured by their \emph{td-error)} receiving from the group. Although they make no direct reference to price of anarchy, the stated goal is to shift the system's equilibrium towards an outcome that maximizes social welfare. In contrast to~\citep{hostallero2020inducing}, D3C agents learn to share varying rewards with individual agents rather than sharing an average gift with everyone in their predefined peer group. This is important as the latter prevents the possible discovery of teams as demonstrated by D3C in~Appx.~\ref{election}.

\citet{yang2020learning} propose an algorithm LIO (Learning to Incentivize Others) that equips agents with ``gifting'' policies represented as neural networks. At each time step, each agent observes the environment and actions of all other agents to determine how much reward to gift to the other agents. The parameters of these networks are adjusted to maximize the original environment reward (without gifts) minus some penalty regularizer for gifting meant to approximately maintain \emph{budget-balance}. In order to perform this maximization, each agent requires access to every other agent's action-policy, gifting-policy, and return making this approach difficult to scale and decentralize. \citet{yang2020learning} demonstrate LIO's ability to maximize welfare and achieve division of labor on a very restricted version of the Cleanup game we evaluate in Appx.~\ref{cleanup}. We also evaluate D3C on this restricted variant in the Appx.~\ref{lio}.

Inspired by social psychology, \citet{mckeesocial} explored imbuing agents with a predisposed \emph{social value orientation} that modifies their rewards. Populations with heterogeneous populations achieved higher fitness scores than homogeneous ones in an evolutionary training approach (i.e., learning occurs outside the agent's lifetime).

One key innovation of D3C beyond the above works is its budget-balance guarantee. In~\citep{hostallero2020inducing,yang2020learning,mckeesocial}, agents manifest extra reward to gift to peers, but no explanation is given for where this extra reward might come from. Also, none of these works tie their proposed approaches to the fundamental game theoretic concept price of anarchy. The derivation of D3C from first principles provides an explicit link, showing an agent-centric learning rule can be approximately consistent with the global objective of maximal social welfare.

Like D3C, OpenAI Five~\citep{openai2019dota} also linearly mixed agents rewards which each other, but where the single ``team spirit" mixture parameter ($\tau$) is \textbf{manually} annealed throughout training from $0.3$ to $1.0$ (i.e., $A_{ii}=1-0.8\tau, A_{ij}=0.2 \tau, j \ne i$).

Finally, we point out that loss transformation is consistent with human behavior. Within social psychology, \textit{interdependence theory}~\citep{kelley1978interpersonal} holds that humans make decisions based on self interest \emph{and} social preferences, allowing them to avoid poor Nash equilibria.

\section{Experiments}
\label{experiments}
Here, we show that agents minimizing local estimates of price of anarchy achieve lower loss on average than selfish, rational agents in five domains.
In the first two domains, a traffic network (4 players) and a generalized prisoner's dilemma (10 players), players optimize using exact gradients (see \eqref{poa_additive}). Then in three RL domains\textemdash Trust-Your-Brother, Coins and Cleanup\textemdash players optimize with approximate gradients as handled by Algorithm~\ref{alg_rl_top}. Agents train with deep networks and A2C~\citep{espeholt2018impala}. We refer to both algorithms as D3C (\uline{decentralized}, \uline{differentiable}, \uline{dynamic} \uline{compromise}).

For D3C, we initialize $A_{ii}=0.99$ and $A_{ij}=\frac{0.01}{n-1}, \, j \ne i$. We initialize away from a onehot because we use entropic mirror descent~\citep{beck2003mirror} to update $A_i$, and this method requires iterates to be initialized to the interior of the simplex. In the RL domains, updates to $A_i$ are clipped in logit-space to be within $l=-5$ and $h=5$ (see Algorithm~\ref{alg_rl_top}). We set the $\mathcal{D}_{KL}$ coefficient to $0$ except for in Coins, where $\nu = 10^{-5}$. Additional hyperparameters are specified in~Appx.~\ref{agents}. In experiments where we cannot compute price of anarchy (\eqref{global_poa_util_def}) exactly, we either report the total loss of the learning algorithm (e.g., D3C) along with the loss achieved by fully cooperative agents ($A_{ij} = \frac{1}{n}$) or the ratio of these losses referred to as ``ratio to optimal''.


\subsection{Traffic Networks and Braess's Paradox}
\label{braess}
%
\begin{figure}[ht!]
    \begin{subfigure}[b]{.5\textwidth}
        \centering
        \begin{minipage}{.25\textwidth}
            \includegraphics[scale=0.45]{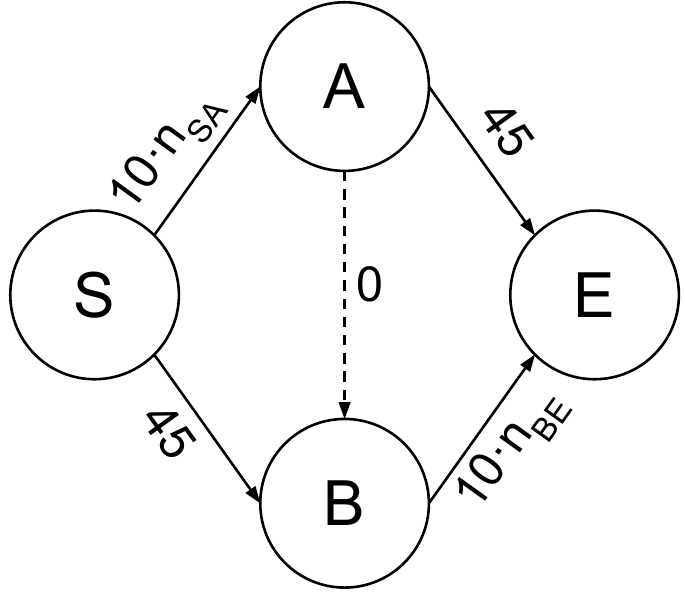}
        \end{minipage}
        \hspace{1.0cm}
        \begin{minipage}{.25\textwidth}
            \begin{align}
                &n_{SA} \in \{0-4\}, \, n_{BE} \in \{0-4\}  \nonumber \\
                &10 n_{SA} + 10 n_{BE} < 10 n_{SA} + 45 \nonumber \\
                &10 n_{SA} + 10 n_{BE} < 10 n_{BE} + 45 \nonumber
            \end{align}
        \end{minipage}
        \caption{Traffic Network \label{fig:traffic_network}}
    \end{subfigure}
    \begin{subfigure}[b]{0.5\textwidth}
        \centering
        \includegraphics[scale=0.28]{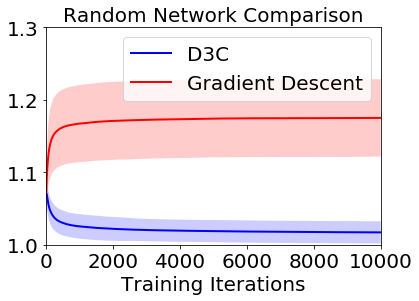}
        \caption{Random Network Results \label{fig:random_traffic_network}}
    \end{subfigure}
    \vspace{-5pt}
    \caption{(\subref{fig:traffic_network}) Four drivers aim to minimize commute time from S to E. Commute time on each edge depends on the number of commuters, $n_{ij}$. Without edge AB, drivers distribute evenly across SAE and SBE for a 65 min commute. After edge AB is added, switching to the shortcut, SABE, always decreases commute time given the other drivers maintain their routes, however, all drivers are incentivized to take the shortcut resulting in an 80 min commute. (\subref{fig:random_traffic_network}) The mean ``ratio to optimal'' over training for 1000 randomly generated networks exhibiting Braess's paradox with $\pm 1$ stdev shaded.}
    \label{fig:braess_diagram_intro}
\end{figure}

In 2009, New York city's mayor closed Broadway near Times Square to alleviate traffic congestion~\citep{neumanbarbaro2009}. This counter-intuitive phenomenon, where restricting commuter choices improves outcomes, is called Braess's paradox~\citep{wardrop1952road,beckmann1956studies,braess1968paradoxon}, and has been observed in real traffic networks~\citep{youn2008price,steinberg1983prevalence}.
%
Braess's paradox is also found in physics~\citep{youn2008price}, decentralized energy grids~\citep{witthaut2012braess}, and can cause extinction cascades in ecosystems~\citep{sahasrabudhe2011rescuing}. Knowing when a network may exhibit this paradox is difficult, which means knowing when network dynamics may result in poor outcomes is difficult.
%

Figure~\ref{fig:braess_diagram_intro}\subref{fig:traffic_network} presents a theoretical traffic network. Without edge AB, drivers commute according to the Nash equilibrium, either learned by gradient descent or D3C. Figure~\ref{fig:braess_results_statistics_tworoad}\subref{fig:coin_stats} shows the price of anarchy approaching 1 for both algorithms.
%
%
%
%
If edge AB is added, the network now exhibits Braess's paradox. Figure~\ref{fig:braess_results_statistics_tworoad}\subref{fig:w_shortcut} shows that while gradient descent converges to Nash ($\rho=\frac{80}{65}$), D3C achieves an average ``ratio to optimal'' near $1$.
\begin{figure}[ht!]
    \begin{subfigure}[b]{.5\textwidth}
        \centering
        \includegraphics[width=0.45\textwidth]{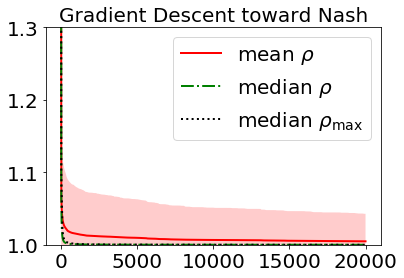}
        \includegraphics[width=0.45\textwidth]{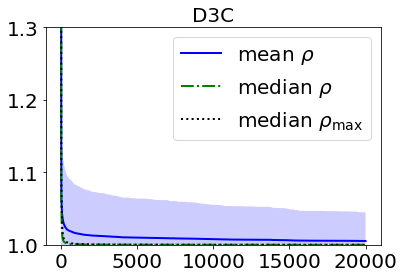}
        \caption{Without Shortcut (Edge AB Removed) \label{fig:wo_shortcut}}
    \end{subfigure}
    \begin{subfigure}[b]{.5\textwidth}
        \centering
        \includegraphics[width=0.45\textwidth]{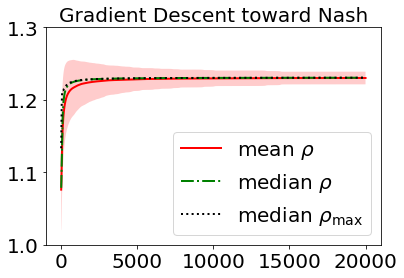}
        \includegraphics[width=0.45\textwidth]{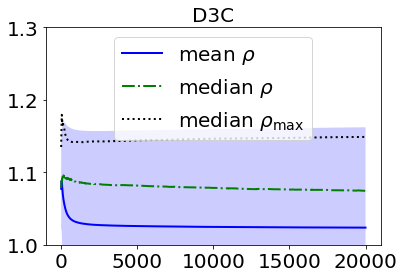}
        \caption{With Shortcut (Edge AB Included) \label{fig:w_shortcut}}
    \end{subfigure}
    \vspace{-5pt}
    \caption{\textbf{Traffic Network}\textemdash (\subref{fig:wo_shortcut}) Without edge AB, agents are initialized with random strategies and train with either gradient descent (left) or D3C (right)\textemdash similar performance is expected. Statistics of $1000$ runs are plotted over training. Median $\rho_{\max}$ tracks the median over trials of the longest-commute among the four drivers. The shaded region captures $\pm$ $1$ stdev around the mean. (\subref{fig:w_shortcut}) After edge AB is added, agents are initialized with random strategies and trained with either gradient descent (left) or D3C (right).}
    \label{fig:braess_results_statistics_tworoad}
\end{figure}
%
%
Figure~\ref{fig:braess_diagram_intro}\subref{fig:random_traffic_network} shows that when faced with a randomly drawn network, D3C agents achieve shorter commutes on average than agents without the ability to compromise.


\subsection{Prisoner's Dilemma}
In an $n$-player prisoner's dilemma, each player must decide to defect or cooperate with each of the other players creating a combinatorial action space of size $2^{n-1}$. This requires a payoff tensor with $2^{n(n-1)}$ entries. Instead of generalizing prisoner's dilemma~\citep{rapoport1965prisoner} to $n$ players using $n$th order tensors, we translate it to a game with convex loss functions.
Figure~\ref{fig:pd_table_example} shows how we can accomplish this.
\begin{figure}[ht]
    \begin{subfigure}[b]{.23\textwidth}
        \centering
        \includegraphics[width=\textwidth]{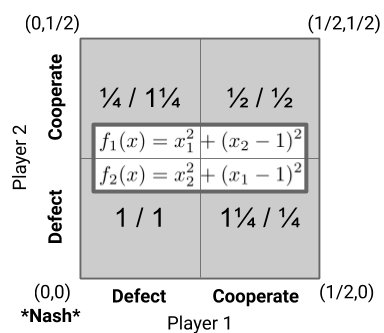}
        \caption{Prisoner's Dilemma \label{fig:pd_table_example}}
    \end{subfigure}
    \begin{subfigure}[b]{.23\textwidth}
        \centering
        \includegraphics[scale=0.3]{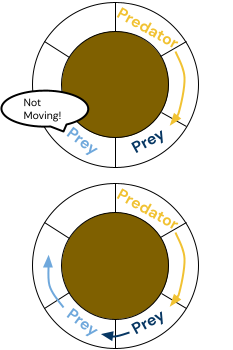}
        \caption{Trust-Your-Brother \label{fig:trustyourbro_game_visual}}
    \end{subfigure}
    \label{fig:domain_illustrations}
    \caption{(\subref{fig:pd_table_example}) A reformulation of the prisoner's dilemma using convex loss functions instead of a normal form payoff table. (\subref{fig:trustyourbro_game_visual}) A bot chases two agents around a table. The predator's prey can only escape if the other prey simultaneously moves out of the way. Selfish (top), cooperative (bottom).}
\end{figure}
Generalizing this to $n$ players, we say that for all $i, j, k$ distinct, 1) player $i$ wants to defect against player $j$, 2) player $i$ wants player $j$ to defect against player $k$, and 3) player $i$ wants player $j$ to cooperate with itself. In other words, each player desires a free-for-all with the exception that no one attacks it. See~Appx.~\ref{pd_convex} for more details.

For three players, we can define the vector of loss functions with
\begin{align}
    \boldsymbol{f}(\boldsymbol{x}) &= \sum_{columns} \Big[ \Big( \begin{bmatrix}
    \boldsymbol{x}^\top \\ \boldsymbol{x}^\top \\ \boldsymbol{x}^\top
    \end{bmatrix}
    - C \Big)^2 \Big]
\end{align}
where $\boldsymbol{x} = [x_{ij}]$ is a column vector ($i \in [1,n], j \in [1,n-1]$) containing the players' (randomly initialized) strategies and $C$ is an $n \times n(n-1)$ matrix with entries that either equal $0$ or $c \in \mathbb{R}^+$.

Figure~\ref{fig:pd_results_statistics} shows that D3C with a randomly initialized strategy successfully minimizes the price of anarchy. In contrast, gradient descent learners provably converge to Nash at the origin with $\rho=\frac{n}{c(n-1)}$. The price of anarchy grows unbounded as $c \rightarrow 0^+$. We set $n=10$ and $c=1$ ($\rho=\frac{10}{9}$) in this experiment with additional settings explored in Appx.~\ref{appx:more_pd}.
\begin{figure}[ht!]
    \centering
    \includegraphics[width=0.23\textwidth]{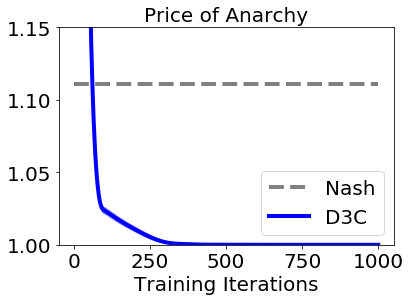}
    \includegraphics[width=0.23\textwidth]{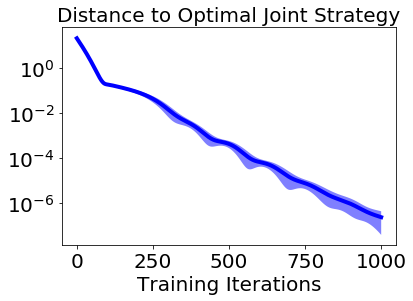}
    \caption{\textbf{Prisoner's Dilemma}\textemdash Convergence to $\rho=1$ (left) and the unique optimal joint strategy (right) over $1000$ runs. The shaded region captures $\pm$ $1$ standard deviation around the mean (too small to see on left). Gradient descent (not shown) provably converges to Nash.}
    \label{fig:pd_results_statistics}
\end{figure}

Figure~\ref{fig:pd_results_singlerun} highlights a single training run. Both agents are initialized to minimize their original loss, but then learn over training to minimize the mean of the two player losses.
\begin{figure}[ht!]
    \centering
    \includegraphics[width=0.23\textwidth]{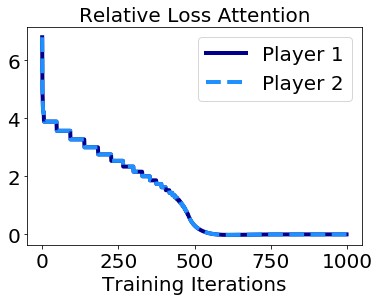}
    \includegraphics[width=0.23\textwidth]{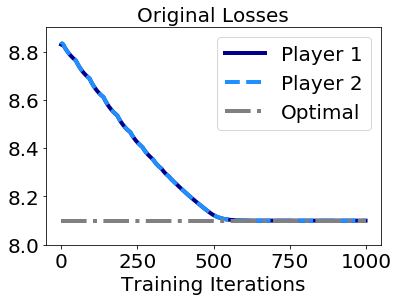}
    \caption{\textbf{Prisoner's Dilemma}\textemdash Single run: relative loss attention measured as $\ln\big(\frac{A_{ii}}{A_{j\ne i}}\big)$ (left) and player losses, $f_i$, (right).}
    \label{fig:pd_results_singlerun}
\end{figure}


\subsection{Trust-Your-Brother}
In this game, a predator chases two prey around a table. The predator uses a hard-coded policy to move towards the nearest prey unless it is already adjacent to a prey, in which case it stays put. If the prey are equidistant to the predator, the predator selects its prey at random. The prey receive $0$ reward if they chose not to move and $-.01$ if they attempted to move. They additionally receive $-1$ if the predator is adjacent to them after moving.

The prey employ linear softmax policies (no bias term) and train via REINFORCE~\citep{williams1992simple}. Both prey receive the same $2$-d observation vector. The first feature specifies the counter-clockwise distance to the predator minus the clockwise distance for the dark blue prey. The second feature specifies the same for the light blue prey.

Figure~\ref{fig:trustyourbro_training} shows D3C approaches maximal total return over training; this is achieved by the agents compromising on their original reward incentives and attending to those of the other agent instead.
\begin{figure}[ht!]
    \centering
    \includegraphics[scale=0.28]{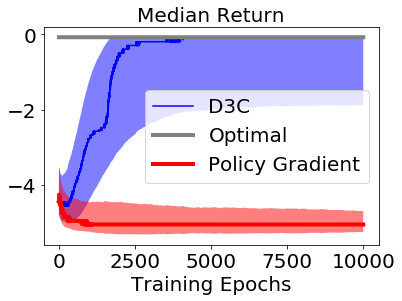}
    \includegraphics[scale=0.28]{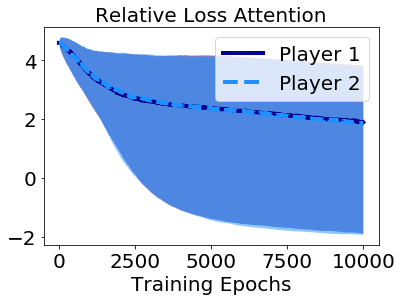}
    \caption{\textbf{Trust-Your-Brother}\textemdash Median return achieved during training for agents trained with policy gradient vs policy gradient augmented with D3C (left); relative reward attention is measured as $\ln\big(\frac{A_{ii}}{A_{j\ne i}}\big)$ where a positive value corresponds to selfish attention and a negative value to other-regarding (right). The $\pm$ $1$ standard deviation shading about the mean for both players overlaps ($1000$ runs).}
    \label{fig:trustyourbro_training}
\end{figure}

\subsection{Coin Dilemma}
In the Coins game~\citep{eccles2019imitation,lerer2017maintaining}, two agents move on a fully-observed $5 \times 5$ gridworld, on which coins of two types corresponding to each agent randomly spawn at each time step with probability $0.005$. When an agent moves into a square with a coin of either type, they get a reward of $1$. When an agent picks up a coin of the other player's type, the other agent receives $-2$. The episode lasts $500$ steps. Total reward is maximized when each agent picks up only coins of their own type, but players are tempted to pick up all coins.
\begin{figure}[ht!]
    \centering
    \begin{subfigure}[b]{.23\textwidth}
    \includegraphics[width=\textwidth]{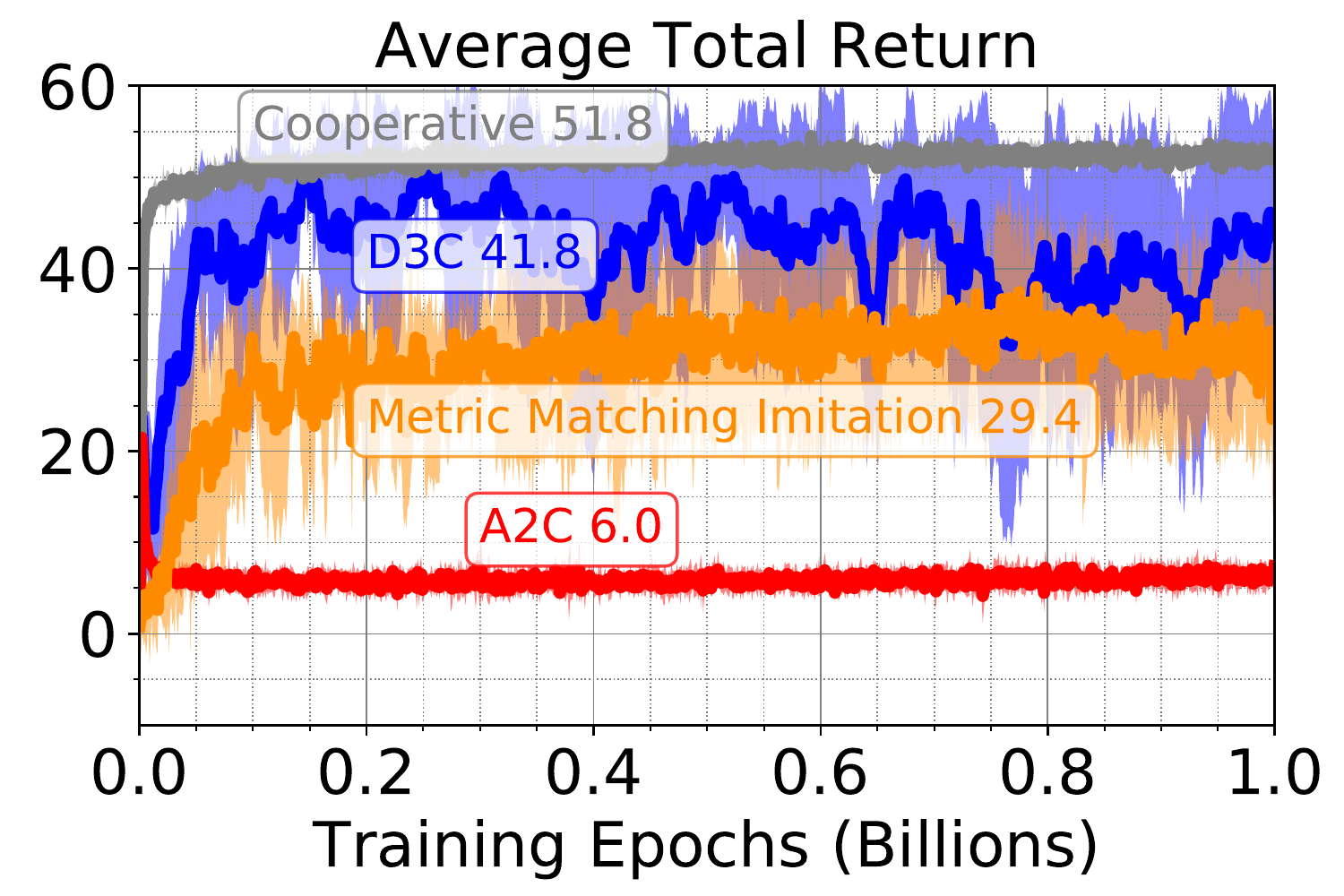}
    \caption{Coin Dilemma: 10 Run Avg \label{fig:coin_stats}}
    \end{subfigure}
    \begin{subfigure}[b]{.23\textwidth}
    \includegraphics[width=\textwidth]{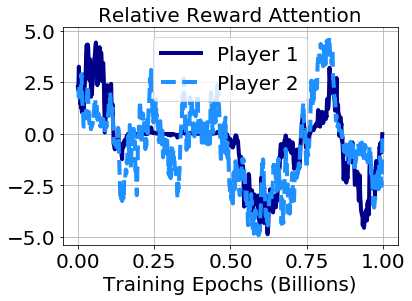}
    \caption{Individual Run: $\ln\big(\frac{A_{ii}}{A_{j\ne i}}\big)$ \label{fig:coin_ind_ratt}}
    \end{subfigure}
    \begin{subfigure}[b]{.23\textwidth}
    \includegraphics[width=\textwidth]{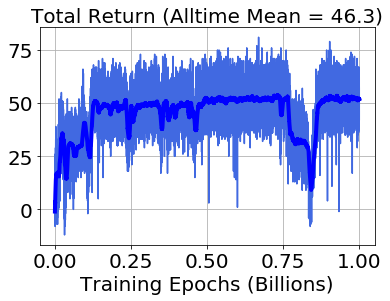}
    \caption{Individual Run: Return \label{fig:coin_ind_r}}
    \end{subfigure}
    \begin{subfigure}[b]{.23\textwidth}
    \includegraphics[width=\textwidth]{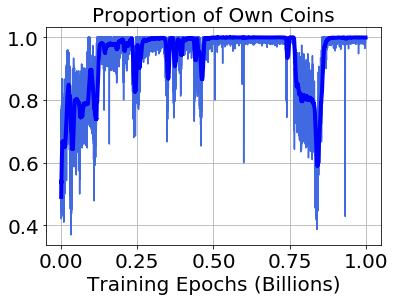}
    \caption{Individual Run: Coins \label{fig:coin_ind_c}}
    \end{subfigure}
    \vspace{-5pt}
    \caption{\textbf{Coin Dilemma}\textemdash (\subref{fig:coin_stats}) Mean total return over ten training runs for agents. Mean return over all epochs is reported in the legend. D3C hyperparameters were selected using five independent validation runs. Cooperative agents trained to maximize total return represent the best possible baseline. Shaded region captures $\pm$ $1$ standard deviation around the mean. (\subref{fig:coin_ind_ratt}-\subref{fig:coin_ind_c}) One training run ($A^0_{ii}=0.9$): relative reward attention measured as $\ln\big(\frac{A_{ii}}{A_{j\ne i}}\big)$ (\subref{fig:coin_ind_ratt}); sum of agent returns (\subref{fig:coin_ind_r}); \% of coins picked up that were the agent's type (\subref{fig:coin_ind_c}).}
    \label{fig:coin_dilemma_stats}
\end{figure}

D3C agents approach optimal cooperative returns (see Figure~\ref{fig:coin_dilemma_stats}\subref{fig:coin_stats}). We compare against Metric Matching Imitation~\citep{eccles2019learning}, which was previously tested on Coins and designed to exhibit reciprocal behavior towards co-players.
Figure~\ref{fig:coin_dilemma_stats}\subref{fig:coin_ind_ratt} shows D3C agents learning to cooperate, then temporarily defecting before rediscovering cooperation. Note that the relative reward attention of both players spikes towards selfish during this small defection window; agents collect more of their opponent's coins during this time. Oscillating between cooperation and defection occurred across various hyperparameter settings. Relative reward attention trajectories between agents appear to be reciprocal (see~Appx.~\ref{coin_reciprocity} for analysis).

\subsection{Cleanup}
\label{cleanup}
We provide additional results on Cleanup, a five-player gridworld game~\citep{hughes2018inequity}. Agents are rewarded for eating apples, but must keep a river clean to ensure apples receive sufficient nutrients. The option to freeload and only eat apples presents a social dilemma. D3C increases both welfare and individual reward over A2C (no loss mixing). We also observe that direct welfare maximization (Cooperation) always results in three agents collecting rewards from apples while two agents sacrifice themselves and clean the river. In contrast, D3C avoids this stark division of labor. Agents take turns on each task and all achieve some positive cumulative return.
\begin{figure}[ht!]
    \begin{subfigure}[b]{.23\textwidth}
        \centering
        \includegraphics[width=\textwidth]{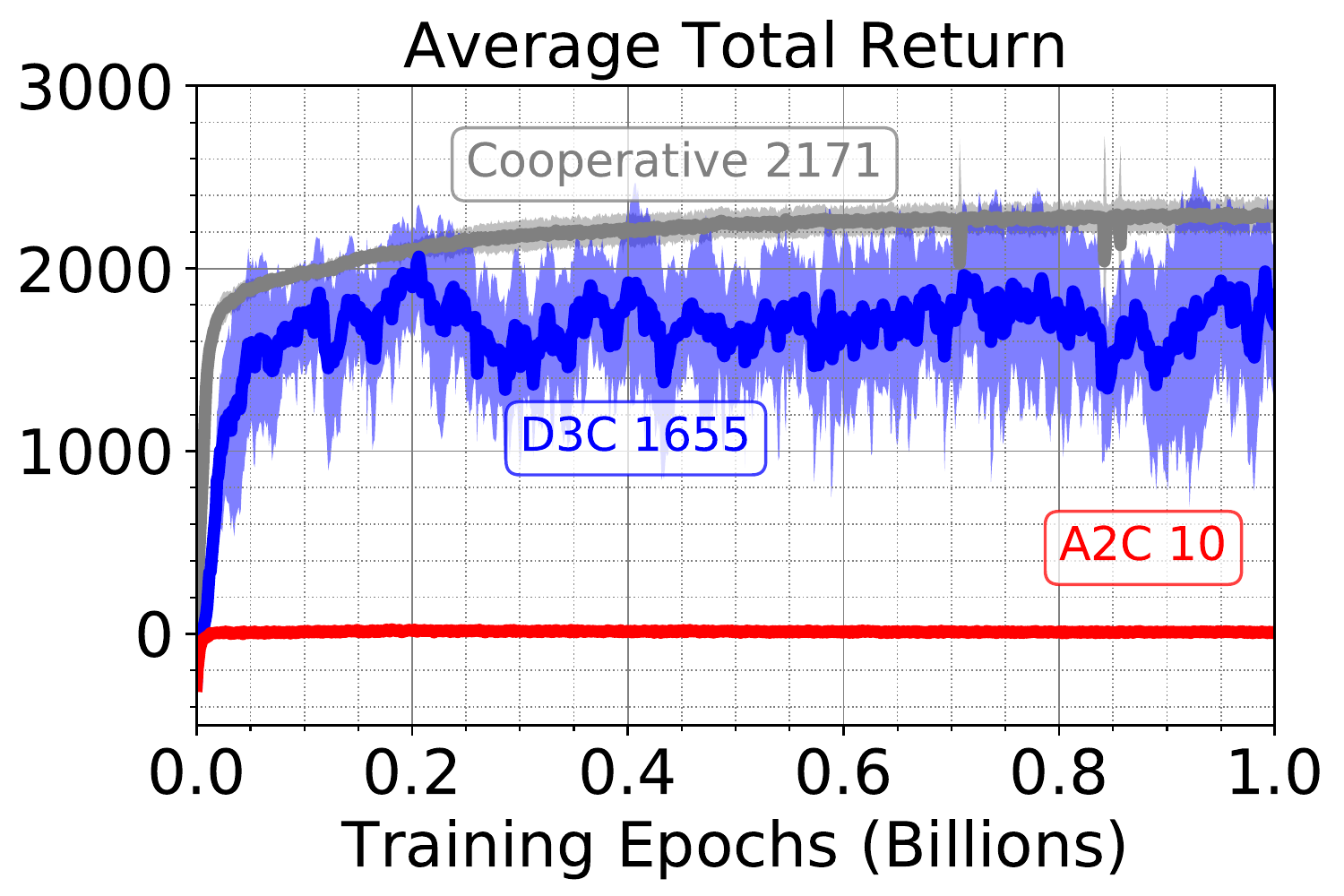}
        \caption{Cleanup: 10 Run Stats~\label{fig:huangpu}}
    \end{subfigure}
    \begin{subfigure}[b]{.23\textwidth}
        \centering
        \includegraphics[width=\textwidth]{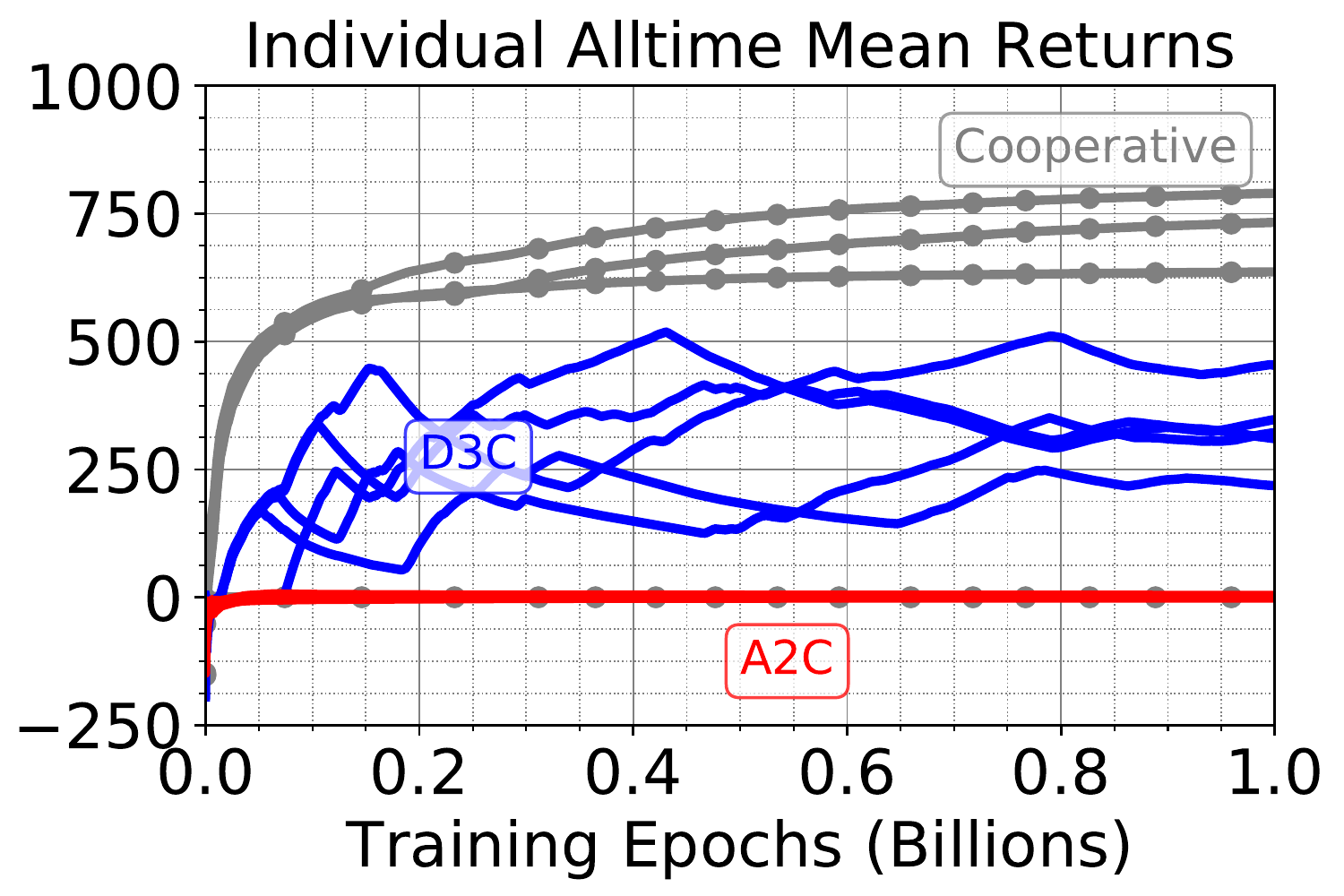}
        \caption{Individual Run~\label{fig:huangpu_ind_1}}
    \end{subfigure}
    \begin{subfigure}[b]{.23\textwidth}
        \centering
        \includegraphics[width=\textwidth]{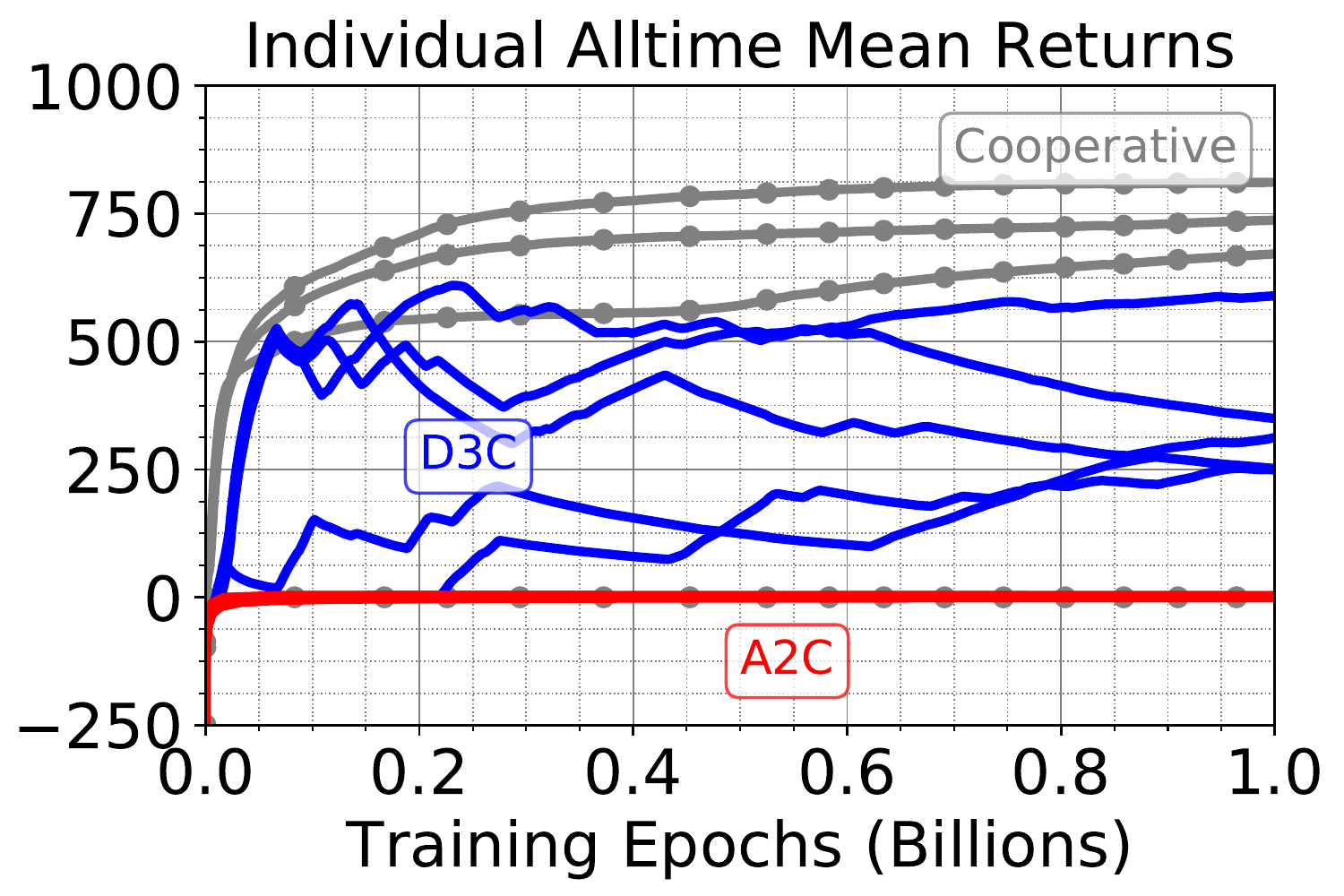}
        \caption{Individual Run~\label{fig:huangpu_ind_2}}
    \end{subfigure}
    \begin{subfigure}[b]{.23\textwidth}
        \centering
        \includegraphics[width=\textwidth]{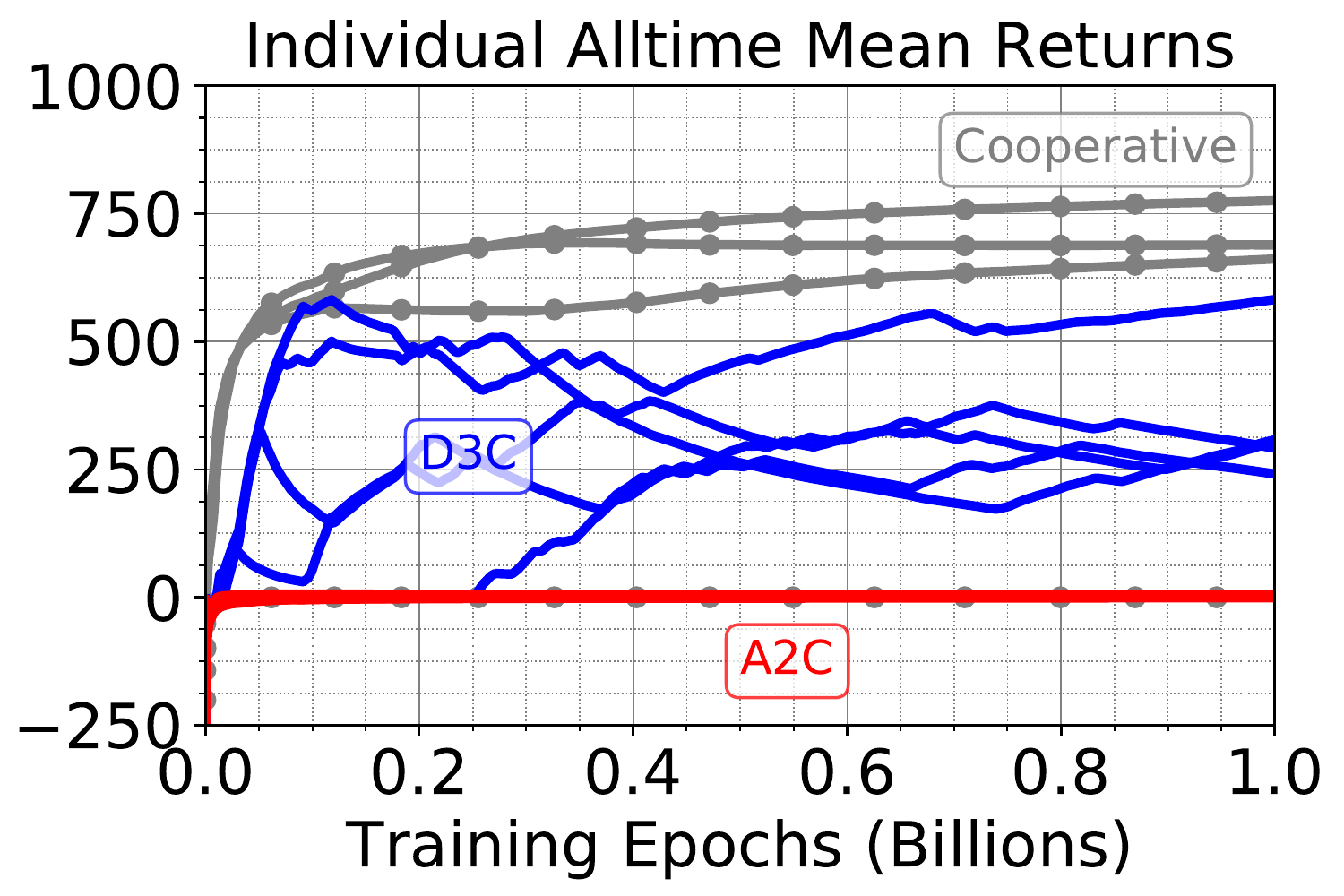}
        \caption{Individual Run~\label{fig:huangpu_ind_3}}
    \end{subfigure}
    \vspace{-5pt}
    \caption{Cleanup (\subref{fig:huangpu}) Mean total returns over ten training runs. D3C hyperparameters were selected using five independent validation runs. Cooperative agents trained to maximize total return represent the best possible baseline. Shaded region captures $\pm$ $1$ standard deviation around the mean. (\subref{fig:huangpu_ind_1}-\subref{fig:huangpu_ind_3}) Three randomly selected runs. Each curve shows the mean return up to the current epoch for 1 of 5 agents.}
    \label{fig:five_player}
\end{figure}

\section{Conclusion}
We formulate learning incentives as a price of anarchy minimization problem and propose a decentralized, gradient-based approach (D3C) that incrementally adapts agent incentives to the environment at hand. We demonstrate its effectiveness on achieving near-optimal agent outcomes in socially adversarial environments.

It is conceptually possible to scale our approach to very large populations through randomly sharing incentives according to the learned mixture weights or sparse optimization over the simplex~\citep{pilanci2012recovery,kyrillidis2013sparse,li2016methods}, but we leave this challenge to future work.

\begin{acks}
We are grateful to Jan Balaguer for fruitful discussions and advice on revising parts of the manuscript.
\end{acks}

\bibliography{main}
\bibliographystyle{ACM-Reference-Format}

\newpage
\onecolumn
\appendix
\appendixpage
\tableofcontents
\addtocontents{toc}{\protect\setcounter{tocdepth}{2}}
\newpage

\section{Mechanism Design}
Mechanism design prescribes a way for resolving compromise between self-interested agents~\citep{nisan2007algorithmic}. For example, in the VCG mechanism~\citep{clarke1971multipart}, all agents must reveal their incentives to a central coordinator, the \emph{principal}. This mechanism achieves optimal group behavior by taxing each agent appropriately but then ``burns'' the collected payments, failing eliminate all the original inefficiency~\citep{hartline2008optimal,green1979incentives,rothkopf2007thirteen}, i.e., VCG is not strongly \emph{budget-balanced}.

\section{Bad Nash \& Futile Opponent Shaping}
\label{bad_nash}
Here, we present a small two-player game where the Nash equilibrium results in poor outcomes for both agents individually and as a group. We then point out how although an opponent shaping approach would typically be able to manipulate players into avoiding such equilibria, it fails in this specific game. We seek a general algorithm for resolving multiagent dilemmas and so we propose a new solution.

\subsection{Bad Nash}
\textbf{Game 1} (Nash Paradox) $\min_{x_1 \in [0,1]} f_1(x_1,x_2) = x_1^2 + \frac{1}{x_2^2 + \kappa}, \,\, \min_{x_2 \in [0,1]} f_2(x_1,x_2) = x_2^2 + \frac{1}{x_1^2 + \kappa}$.\\\\
The unique Nash equilibrium of this general-sum game is $(x_1,x_2)$$=$$(0,0)$ regardless of $\kappa \in [0,1)$; at Nash, each player sees a loss of $\frac{1}{\kappa}$. The minimal total loss solution is $(x_1,x_2) = (\sqrt{1-\kappa}, \sqrt{1-\kappa})$ for $\kappa < 1$ where each player sees a loss of $2-\kappa$. The price of anarchy is $\frac{1/\kappa}{2-\kappa}$ which goes to $\infty$ as $\kappa \rightarrow 0$. For $\kappa < $ golden ratio$- 1 \approx 0.618$, Nash achieves maximum total loss among all possible strategy sets. While computing a Nash is an important technical problem, Game~1 proves that even if a Nash can be computed, it may be undesirable. Thus solving for Nash is orthogonal to this work.

\subsection{Gradient Descent Without Descent}
\label{gdwgd}
Game~1 shows that the Nash equilibrium can give the worst outcome for all agents. It follows that agents learning with gradient descent in this game must observe their loss increase upon their final approach to Nash. Why stick to gradient descent then? In multiagent games, the adjustment of another player's strategy coupled with our own can increase our loss.
%
Let $f_i(t)$ be shorthand for $f_i(\mathbf{x}(t))$ where $\mathbf{x}(t)$ contains all strategies at time (iteration) $t$. Then a series expansion (see \eqref{eqn:series}) of agent $i$'s loss around the current time step makes this concrete:
\begin{align}
    &f_i(t+\Delta t) = f_i(t) + \Delta t \frac{df_i}{dt} + \frac{\Delta t^2}{2} \frac{d^2 f_i}{dt^2} + \mathcal{O}(\Delta t^3) \label{eqn:series} \\
    &= f_i(t) + \Delta t \frac{\partial f_i}{\partial x_i} \frac{dx_i}{dt} + \frac{\Delta t^2}{2} \Big[ \frac{\partial^2 f_i}{\partial x_i^2} \Big( \frac{dx_i}{dt} \Big)^2 + \frac{\partial f_i}{\partial x_i} \frac{d^2 x_i}{dt^2} + 2 \sum_{j \ne i} \boldsymbol{\frac{\partial^2 f_i}{\partial x_i \partial x_j} \frac{dx_i}{dt} \frac{dx_j}{dt}} \Big] \nonumber
    \\ &+ h(\frac{dx_{j \ne i}}{dt}) + \mathcal{O}(\Delta t^3)
\end{align}
where $\Delta t > 0$ is a small learning rate and $h(\frac{dx_{j \ne i}}{dt})$ contains terms that agent $i$ cannot manipulate (i.e., $h$ is constant w.r.t. agent $i$'s update dynamics, $\frac{dx_i}{dt}$). We show the full derivation of the series expansion in Section~\ref{app:gdwgd} for those interested.

\subsection{The Update Is Not The Only Problem}
\label{game_is_prob}
In \eqref{eqn:series}, other agents can affect $f_i(t+\Delta t)$ through the bold terms and $h(\frac{dx_{j \ne i}}{dt})$.
The bold terms indicate where agent $i$'s update couples with other players' updates~\citep{schafer2019competitive}. To account for these terms, agent $i$ must predict the other agents' updates, $\frac{d x_j}{dt}$, and understand how their behaviors affect agent $i$'s loss, $\frac{d^2f_i}{dx_i dx_j}$.
%
Recent methods, such as LOLA, LookAhead and Stable Opponent Shaping~\citep{foerster2018learning,letcher2018stable}, model these terms. However, all these methods converge to Nash in Game~1 because $\frac{d^2f_i}{dx_i dx_j} = 0$ as do all other mixed derivatives of agent $i$'s loss.
%
In contrast, agent $i$ can never mitigate increases in loss due to $h$. Incorporating more terms in the expansion generates higher level reasoning, but even the infinite expansion cannot avoid the \emph{Nash paradox} in Game~1. If $x_1$ knows $x_2$'s learning trajectory converges to $0$, $x_1$ is still incentivized to play $0$. \textbf{The fault lies in the game, not the learning.}

\section{Taylor Series Expansion}
\label{app:gdwgd}

Here, we derive the Taylor series expansion given in Section~\ref{gdwgd}. The derivation is as follows:
\begin{align}
    \frac{df_i}{dt} &= \sum_j \frac{\partial f_i}{\partial x_j} \frac{dx_j}{dt} = \frac{\partial f_i}{\partial x_i} \frac{dx_i}{dt} + \textcolor{blue}{\sum_{j \ne i} \frac{\partial f_i}{\partial x_j} \frac{dx_j}{dt}}
    \\ \frac{d^2 f_i}{dt^2} &= \frac{d}{dt}\Big( \frac{df_i}{dt} \Big) = \frac{d}{dt} \Big( \sum_j \frac{\partial f_i}{\partial x_j} \frac{dx_j}{dt} \Big) = \sum_j \frac{d}{dt} \Big( \frac{\partial f_i}{\partial x_j} \frac{dx_j}{dt} \Big)
    \\ &= \sum_j \Big[ \frac{d}{dt} \Big( \frac{\partial f_i}{\partial x_j} \Big) \frac{dx_j}{dt} + \frac{\partial f_i}{\partial x_j} \frac{d^2x_j}{dt^2} \Big]
    \\ &= \sum_j \Big[ \Big( \sum_k \frac{\partial^2 f_i}{\partial x_j \partial x_k} \frac{dx_k}{dt} \Big) \frac{dx_j}{dt} + \frac{\partial f_i}{\partial x_j} \frac{d^2x_j}{dt^2} \Big]
    \\ &= 2 \sum_{j \ne i} \frac{\partial^2 f_i}{\partial x_j \partial x_i} \frac{dx_i}{dt} \frac{dx_j}{dt} + \frac{\partial^2 f_i}{\partial x_i^2} \Big( \frac{dx_i}{dt} \Big)^2
    \\ &+ \textcolor{blue}{\sum_{j \ne i} \sum_{k \ne i} \frac{\partial^2 f_i}{\partial x_j \partial x_k} \frac{dx_j}{dt} \frac{dx_k}{dt} + \sum_{j \ne i} \frac{\partial f_i}{\partial x_j} \frac{d^2 x_j}{dt^2}} + \frac{\partial f_i}{\partial x_i} \frac{d^2 x_i}{dt^2}
    \\ \textcolor{blue}{h(\frac{dx_{j \ne i}}{dt})} &= \Delta t \sum_{j \ne i} \frac{\partial f_i}{\partial x_j} \frac{dx_j}{dt} + \frac{\Delta t^2}{2} \Big[ \sum_{j \ne i} \sum_{k \ne i} \frac{\partial^2 f_i}{\partial x_j \partial x_k} \frac{dx_j}{dt} \frac{dx_k}{dt} + \sum_{j \ne i} \frac{\partial f_i}{\partial x_j} \frac{d^2 x_j}{dt^2} \Big]
    \\ f_i(t+\Delta t) &= f_i(t) + \Delta t \frac{\partial f_i}{\partial x_i} \frac{dx_i}{dt} + \frac{\Delta t^2}{2} \Big[ \frac{\partial^2 f_i}{\partial x_i^2} \Big( \frac{dx_i}{dt} \Big)^2 + \frac{\partial f_i}{\partial x_i} \frac{d^2 x_i}{dt^2} + 2 \sum_{j \ne i} \boldsymbol{\frac{\partial^2 f_i}{\partial x_i \partial x_j} \frac{dx_i}{dt} \frac{dx_j}{dt}} \Big] + \textcolor{blue}{h(\frac{dx_{j \ne i}}{dt})} + \mathcal{O}(\Delta t^3).
\end{align}

\section{Derivation of an Upper Bound on \emph{Local} Price of Anarchy}
\label{bound_proofs}

\begin{repdefinition}{def:smooth_game}[Smooth Game]
    A game is $(\lambda, \mu)$-smooth~\citep{roughgarden2015intrinsic} if:
    \begin{align}
        \sum_{i=1}^n f^A_i(x_i,x'_{-i}) &\le \lambda \sum_{i=1}^n f^A_i(x_i,x_{-i}) + \mu \sum_{i=1}^n f^A_i(x'_i,x'_{-i})
    \end{align}
    for all $\boldsymbol{x}, \boldsymbol{x}' \in \mathcal{X}$ where $\lambda > 0$, $\mu < 1$, and $\sum_i f_i^A(\boldsymbol{x})$ is assumed to be non-negative for any $\boldsymbol{x} \in \mathcal{X}$.
\end{repdefinition}
The last condition is needed for the price of anarchy to be meaningful.

\begin{replemma}{lemma:smooth_to_poa}[Smooth Games Imply a Bound on Price of Anarchy]
The price of anarchy, $\rho$, the ratio of the worst case Nash total loss to the minimal total loss, is bounded above by a ratio of the coefficients of a smooth game~\citep{roughgarden2015intrinsic}:
\begin{align}
    \rho &= \frac{\max_{\mathcal{X}^*} \sum_i f^A_i(\boldsymbol{x}^*)}{\min_{\mathcal{X}} \sum_i f^A_i(\boldsymbol{x})} \ge 1 \\
    &\le \inf_{\lambda>0, \mu<1} \Big[ \frac{\lambda}{1-\mu} \Big].
\end{align}
where $\boldsymbol{x}^*$ is an element of the set of Nash equilibria, $\mathcal{X}^*$.
\end{replemma}

Assume the loss function gradients are Lipschitz as well. We say a loss function, $f_i^A(\boldsymbol{x}) = f_i^A(x_i, x_{-i})$, has a $\beta_i$-Lipschitz gradient for all $A$ if
\begin{align}
    ||\nabla_{x_i} f_i^A(\boldsymbol{x}) - \nabla_{y_i} f_i^A(\boldsymbol{y}) || &\le \beta_i ||\boldsymbol{x}-\boldsymbol{y}|| \,\, \forall \boldsymbol{x}, \boldsymbol{y}, A.
\end{align}
Note that this implies
\begin{align}
    ||\nabla_{x_i} f_i^A(x_i,z_{-i}) - \nabla_{y_i} f_i^A(y_i, z_{-i}) || &\le \beta_i ||x_i-y_i|| \,\, \forall x_i, y_i, z_{-i}, A
\end{align}
as a special case.

The following lemmas are useful in deriving a local notion of smoothness.
\begin{lemma}
\label{norm_diff}
If $f_i^A(x_i, x_{-i}) = g_i(\boldsymbol{x})$ has a $\beta_i$-Lipschitz gradient, then
\begin{equation}
\Big\vert||\nabla_{x_i} g_i(\boldsymbol{x})|| - ||\nabla_{y_i} g_i(\boldsymbol{y})||\Big\vert \le \beta_i ||\boldsymbol{x}-\boldsymbol{y}|| \,\, \forall \,\, \boldsymbol{x}, \boldsymbol{y}.
\end{equation}
\end{lemma}
\begin{proof}
The proof proceeds in two main steps. First,
\begin{align}
    ||\nabla_{y_i} g_i(\boldsymbol{y})|| &= ||\nabla_{x_i} g_i(\boldsymbol{x}) + \nabla_{y_i} g_i(\boldsymbol{y}) - \nabla_{x_i} g_i(\boldsymbol{x})|| \\
    &\le ||\nabla_{x_i} g_i(\boldsymbol{x})|| + ||\nabla_{x_i} g_i(\boldsymbol{x})-\nabla_{y_i} g_i(\boldsymbol{y})|| \text{ by triangle inequality} \\
    &\le ||\nabla_{x_i} g_i(\boldsymbol{x})|| + \beta_i ||\boldsymbol{x}-\boldsymbol{y}|| \text{ by Lipschitz gradient}
\end{align}
which implies $||\nabla_{y_i} g_i(\boldsymbol{y})|| - ||\nabla_{x_i} g_i(\boldsymbol{x})|| \le \beta_i ||\boldsymbol{x}-\boldsymbol{y}||$. And vice versa,
\begin{align}
    ||\nabla_{x_i} g_i(\boldsymbol{x})|| &= ||\nabla_{y_i} g_i(\boldsymbol{y}) + \nabla_{x_i} g_i(\boldsymbol{x}) - \nabla_{y_i} g_i(\boldsymbol{y})|| \\
    &\le ||\nabla_{y_i} g_i(\boldsymbol{y})|| + ||\nabla_{x_i} g_i(\boldsymbol{x})-\nabla_{y_i} g_i(\boldsymbol{y})|| \text{ by triangle inequality} \\
    &\le ||\nabla_{y_i} g_i(\boldsymbol{y})|| + \beta_i ||\boldsymbol{x}-\boldsymbol{y}|| \text{ by Lipschitz gradient}
\end{align}
which implies $||\nabla_{x_i} g_i(\boldsymbol{x})|| - ||\nabla_{y_i} g_i(\boldsymbol{y})|| \le \beta_i ||\boldsymbol{x}-\boldsymbol{y}||$. The two implications together prove the lemma.
\end{proof}

\begin{lemma}
\label{norm_sq_diff}
If $f_i^A(x_i, x_{-i}) = g_i(\boldsymbol{x})$ has a $\beta_i$-Lipschitz gradient, then
\begin{equation}
\Big\vert||\nabla_{x_i} g_i(\boldsymbol{x})||^2 - ||\nabla_{y_i} g_i(\boldsymbol{y})||^2\Big\vert \le 3\beta_i^2 ||\boldsymbol{x}-\boldsymbol{y}||^2 + 2\beta_i ||\nabla_{x_i} g_i(\boldsymbol{x})|| ||\boldsymbol{x}-\boldsymbol{y}|| \,\, \forall \,\, x,y.
\end{equation}
\end{lemma}
\begin{proof}
The proof proceeds similarly to before. First,
\begin{align}
    ||\nabla_{y_i} g_i(\boldsymbol{y})||^2 &= ||\nabla_{x_i} g_i(\boldsymbol{x}) + \nabla_{y_i} g_i(\boldsymbol{y}) - \nabla_{x_i} g_i(\boldsymbol{x})||^2 \\
    &\le \Big(||\nabla_{x_i} g_i(\boldsymbol{x})|| + ||\nabla_{x_i} g_i(\boldsymbol{x})-\nabla_{y_i} g_i(\boldsymbol{y})||\Big)^2 \text{ by triangle inequality} \\
    &= ||\nabla_{x_i} g_i(\boldsymbol{x})||^2 + ||\nabla_{x_i} g_i(\boldsymbol{x})-\nabla_{y_i} g_i(\boldsymbol{y})||^2 + 2 ||\nabla_{x_i} g_i(\boldsymbol{x})|| ||\nabla_{x_i} g_i(\boldsymbol{x}) - \nabla_{y_i} g_i(\boldsymbol{y})|| \\
    &\le ||\nabla_{x_i} g_i(\boldsymbol{x})||^2 + \beta_i^2 ||\boldsymbol{x}-\boldsymbol{y}||^2 + 2 \beta_i ||\nabla_{x_i} g_i(\boldsymbol{x})|| ||\boldsymbol{x}-\boldsymbol{y}|| \text{ by Lipschitz gradient and Lemma~\ref{norm_diff}}
\end{align}
which implies $||\nabla_{y_i} g_i(\boldsymbol{y})||^2 - ||\nabla_{x_i} g_i(\boldsymbol{x})||^2 \le \beta_i^2 ||\boldsymbol{x}-\boldsymbol{y}||^2 + 2 \beta_i ||\nabla_{x_i} g_i(\boldsymbol{x})|| ||\boldsymbol{x}-\boldsymbol{y}||$. And vice versa,
\begin{align}
    ||\nabla_{x_i} g_i(\boldsymbol{x})||^2 &= ||\nabla_{y_i} g_i(\boldsymbol{y}) + \nabla_{x_i} g_i(\boldsymbol{x}) - \nabla_{y_i} g_i(\boldsymbol{y})||^2 \\
    &\le \Big(||\nabla_{y_i} g_i(\boldsymbol{y})|| + ||\nabla_{x_i} g_i(\boldsymbol{x})-\nabla_{y_i} g_i(\boldsymbol{y})||\Big)^2 \text{ by triangle inequality} \\
    &= ||\nabla_{y_i} g_i(\boldsymbol{y})||^2 + ||\nabla_{x_i} g_i(\boldsymbol{x})-\nabla_{y_i} g_i(\boldsymbol{y})||^2 + 2 ||\nabla_{y_i} g_i(\boldsymbol{y})|| ||\nabla_{x_i} g_i(\boldsymbol{x}) - \nabla_{y_i} g_i(\boldsymbol{y})|| \\
    &\le ||\nabla_{y_i} g_i(\boldsymbol{y})||^2 + \beta_i^2 ||\boldsymbol{x}-\boldsymbol{y}||^2 + 2 \beta_i ||\nabla_{y_i} g_i(\boldsymbol{y})|| ||\boldsymbol{x}-\boldsymbol{y}|| \text{ by Lipschitz gradient and Lemma~\ref{norm_diff}}
\end{align}
which implies $||\nabla_{x_i} g_i(\boldsymbol{x})||^2 - ||\nabla_{y_i} g_i(\boldsymbol{y})||^2 \le \beta_i^2 ||\boldsymbol{x}-\boldsymbol{y}||^2 + 2 \beta_i ||\nabla_{y_i} g_i(\boldsymbol{y})|| ||\boldsymbol{x}-\boldsymbol{y}||$. The two implications together imply
\begin{align}
    \Big\vert||\nabla_{x_i} g_i(\boldsymbol{x})||^2 - ||\nabla_{y_i} g_i(\boldsymbol{y})||^2\Big\vert &\le \beta_i^2 ||\boldsymbol{x}-\boldsymbol{y}||^2 + 2 \beta_i \max\{||\nabla_{x_i} g_i(\boldsymbol{x})||,||\nabla_{y_i} g_i(\boldsymbol{y})||\} ||\boldsymbol{x}-\boldsymbol{y}|| \\
    &\le \beta_i^2 ||\boldsymbol{x}-\boldsymbol{y}||^2 + 2 \beta_i \max\{||\nabla_{x_i} g_i(\boldsymbol{x})||,||\nabla_{x_i} g_i(\boldsymbol{x})|| + \beta_i||\boldsymbol{x}-\boldsymbol{y}||\} ||\boldsymbol{x}-\boldsymbol{y}|| \\
    &= 3 \beta_i^2 ||\boldsymbol{x}-\boldsymbol{y}||^2 + 2 \beta_i ||\nabla_{x_i} g_i(\boldsymbol{x})|| ||\boldsymbol{x}-\boldsymbol{y}||
\end{align}
where the last inequality follows from Lemma~\ref{norm_diff}.
\end{proof}

\begin{lemma}
\label{smooth_1}
If $f_i^A(x_i, x_{-i}) = g_i(\boldsymbol{x})$ has a $\beta_i$-Lipschitz gradient, then there exists a $\Delta t > 0$ sufficiently small s.t.
\begin{equation}
\langle \nabla_{x_i} g_i(\boldsymbol{x}), \nabla_{x_i'} g_i(\boldsymbol{x}') \rangle \ge ||\nabla_{x_i} g_i(\boldsymbol{x})||^2 - \delta_i \Delta t - \gamma_i \Delta t^2 \ge 0 
\end{equation}
where $x_i' = x_i - \Delta t \nabla_{x_i} g_i(\boldsymbol{x})$ for each $i$, $\Delta t > 0$, $\delta_i = \beta_i ||\nabla_{x_i} g_i(\boldsymbol{x})|| \zeta$, $\gamma_i = 2 \beta_i^2 \zeta^2$, and $\zeta = \sqrt{\sum_j ||\nabla_{x_j} g_j(\boldsymbol{x})||^2}$.
\end{lemma}
\begin{proof}
We begin with the assumption of a Lipschitz gradient which trivially implies the following:
\begin{align}
    ||\nabla_{x_i} g_i(\boldsymbol{x}) - \nabla_{y_i} g_i(\boldsymbol{y}) || &\le \beta_i ||\boldsymbol{x}-\boldsymbol{y}|| \,\, \forall \boldsymbol{x}, \boldsymbol{y} \\
    \implies ||\nabla_{x_i} g_i(\boldsymbol{x}) - \nabla_{y_i} g_i(\boldsymbol{y}) ||^2 &\le \beta_i^2 ||\boldsymbol{x}-\boldsymbol{y}||^2 \,\, \forall \boldsymbol{x}, \boldsymbol{y}.
\end{align}
This, in turn, is equivalent to
\begin{align}
    \langle \nabla_{x_i} g_i(\boldsymbol{x}) - \nabla_{y_i} g_i(\boldsymbol{y}), \nabla_{x_i} g_i(\boldsymbol{x}) - \nabla_{y_i} g_i(\boldsymbol{y}) \rangle &\le \beta_i^2 ||\boldsymbol{x}-\boldsymbol{y}||^2 \,\, \forall \boldsymbol{x}, \boldsymbol{y} \\
    =||\nabla_{x_i} g_i(\boldsymbol{x})||^2 + ||\nabla_{y_i} g_i(\boldsymbol{y})||^2 - 2\langle \nabla_{x_i} g_i(\boldsymbol{x}), \nabla_{y_i} g_i(\boldsymbol{y}) \rangle &\le \beta_i^2 ||\boldsymbol{x}-\boldsymbol{y}||^2 \,\, \forall \boldsymbol{x}, \boldsymbol{y}.
\end{align}
Rearranging terms gives
\begin{align}
    \langle \nabla_{x_i} g_i(\boldsymbol{x}), \nabla_{y_i} g_i(\boldsymbol{y}) \rangle &\ge \frac{1}{2}\Big[ ||\nabla_{x_i} g_i(\boldsymbol{x})||^2 + ||\nabla_{y_i} g_i(\boldsymbol{y})||^2 - \beta_i^2 ||\boldsymbol{x}-\boldsymbol{y}||^2 \Big] \,\, \forall \boldsymbol{x}, \boldsymbol{y}.
\end{align}
Now let $y_i= x_i' = x_i - \Delta t \nabla_{x_i} g_i(\boldsymbol{x})$ for each $i$. Lemma~\ref{norm_sq_diff} implies
\begin{align}
    ||\nabla_{x'_i} g_i(\boldsymbol{x}')||^2 &\ge ||\nabla_{x_i} g_i(\boldsymbol{x})||^2 - 3\beta_i^2||\boldsymbol{x}-\boldsymbol{y}||^2 - 2\beta_i ||\nabla_{x_i} g_i(\boldsymbol{x})|| ||\boldsymbol{x}-\boldsymbol{y}|| \\
    &= ||\nabla_{x_i} g_i(\boldsymbol{x})||^2 - 3\beta_i^2 \Delta t^2 \underbrace{\sum_j ||\nabla_{x_j} g_j(\boldsymbol{x})||^2}_{\zeta^2} - 2\beta_i \Delta t ||\nabla_{x_i} g_i(\boldsymbol{x})|| \underbrace{\sqrt{\sum_j ||\nabla_{x_j} g_j(\boldsymbol{x})||^2}}_{\zeta}.
\end{align}

Then
\begin{align}
    \langle \nabla_{x_i} g_i(\boldsymbol{x}), \nabla_{x_i'} g_i(\boldsymbol{x}') \rangle &\ge \frac{1}{2}\Big[ ||\nabla_{x_i} g_i(\boldsymbol{x})||^2 + ||\nabla_{x'_i} g_i(\boldsymbol{x}')||^2 - \Delta t^2 \beta_i^2 \sum_j ||\nabla_{x_j} g_j(\boldsymbol{x})||^2 \Big] \\
    &\ge \frac{1}{2}\Big[ 2||\nabla_{x_i} g_i(\boldsymbol{x})||^2 - 2\beta_i \Delta t ||\nabla_{x_i} g_i(\boldsymbol{x})|| \zeta - 3 \Delta t^2 \beta_i^2 \zeta^2 - \Delta t^2 \beta_i^2 \zeta^2 \Big] \\
    &= ||\nabla_{x_i} g_i(\boldsymbol{x})||^2 - \Delta t \beta_i ||\nabla_{x_i} g_i(\boldsymbol{x})|| \zeta - 2 \Delta t^2 \beta_i^2 \zeta^2 \\
    &= ||\nabla_{x_i} g_i(\boldsymbol{x})||^2 - \delta_i \Delta t - \gamma_i \Delta t^2
\end{align}
where $\delta_i = \beta_i ||\nabla_{x_i} g_i(\boldsymbol{x})|| \zeta$ and $\gamma_i = 2 \beta_i^2 \zeta^2$. Note that $||\nabla_{x_i} g_i(\boldsymbol{x})||^2 \ge 0$ and if $||\nabla_{x_i} g_i(\boldsymbol{x})||^2 = 0$, then $\langle \nabla_{x_i} g_i(\boldsymbol{x}), \nabla_{x_i'} g_i(\boldsymbol{x}') \rangle = 0$.
\end{proof}

\begin{lemma}
\label{smooth_2}
If $f_i^A(x_i, x_{-i}) = g_i(\boldsymbol{x})$ has a $\beta_i$-Lipschitz gradient, then
\begin{equation}
f_i^A(x_i, x'_{-i}) \ge f_i^A(x_i', x'_{-i}) + \langle \nabla_{x'_i} f_i^A(x'_i, x'_{-i}), x_i-x'_i \rangle - \frac{\beta_i}{2} ||x_i-x'_i||^2 
\end{equation}
where $x_i' = x_i - \Delta t \nabla_{x_i} g_i(\boldsymbol{x})$ for each $i$ and $\Delta t > 0$.
\end{lemma}
\begin{proof}
Let $f_i^A(x_i, x'_{-i}) = h_i(x_i)$. We begin with the assumption of a Lipschitz gradient which implies the following:
\begin{align}
    ||\nabla_{x_i} h_i(x_i) - \nabla_{y_i} h_i(y_i) || &\le \beta_i ||x_i-y_i|| \,\, \forall \boldsymbol{x}, \boldsymbol{y} \\
    \implies |h_i(x_i) - h_i(y_i) - \langle \nabla_{y_i} h_i(y_i), x_i-y_i \rangle | &\le \frac{\beta_i}{2} ||x_i-y_i||^2 \,\, \forall \boldsymbol{x}, \boldsymbol{y}.
\end{align}
This then implies
\begin{align}
    h_i(x_i) = h_i(y_i) + \langle \nabla_{y_i} h_i(y_i), x_i-y_i \rangle &+ \kappa_i ||x_i-y_i||^2  \,\, \forall \boldsymbol{x}, \boldsymbol{y} \text{ where $\kappa_i \in [-\frac{\beta_i}{2}, \frac{\beta_i}{2}]$}
\end{align}
Rewriting with $f_i^A$ for clarity, letting $y_i = x_i' = x_i - \Delta t \nabla_{x_i} g_i(\boldsymbol{x})$ for each $i$, and selecting the lower bound gives
\begin{align}
    f_i^A(x_i, x'_{-i}) \ge f_i^A(x_i', x'_{-i}) + \langle \nabla_{x'_i} f_i^A(x'_i, x'_{-i}), x_i-x'_i \rangle &- \frac{\beta_i}{2} ||x_i-x'_i||^2.
\end{align}
\end{proof}

\begin{lemma}
\label{taylor}
If every $f_i^A(x_i, x_{-i}) = g_i(\boldsymbol{x})$ has a $\beta_i$-Lipschitz gradient, then by Lemmas~\ref{smooth_1} and~\ref{smooth_2}, there exists a $\Delta t$ such that
\begin{equation}
    \sum_{i=1}^n f^A_i(x_i,x'_{-i}) \ge \sum_{i=1}^n f^A_i(x'_i,x'_{-i}) + \underbrace{a_i}_{\ge 0}
\end{equation}
where $x_i' = x_i - \Delta t \nabla_{x_i} f^A_i(\boldsymbol{x})$ and $a_i = ||\nabla_{x_i} f_i^A(x_i, x_{-i})||^2 - \delta_i \Delta t - \gamma_i \Delta t^2$ for each $i$.
\end{lemma}
\begin{proof}
Consider simultaneous gradient descent dynamics. Let $x_i' = x_i - \Delta t \nabla_{x_i} f^A_i(\boldsymbol{x})$. Then by Lemmas~\ref{smooth_1} and~\ref{smooth_2}, we find
\begin{align}
    f^A_i(x_i,x'_{-i}) &\ge f^A_i(x'_i,x'_{-i}) + \langle \nabla_{x'_i} f^A_i(x'_i, x'_{-i}), x_i - x'_i \rangle - \frac{\beta_i}{2} ||x_i-x'_i||^2 \\
    &= f^A_i(x'_i,x'_{-i}) + \Delta t \langle \nabla_{x'_i} f^A_i(x'_i, x'_{-i}), \nabla_{x_i} f^A_i(x_i, x_{-i}) \rangle - \frac{\beta_i}{2} ||x_i-x'_i||^2 \\
    &= f^A_i(x'_i,x'_{-i}) + \Delta t \langle \nabla_{x'_i} f^A_i(x'_i, x'_{-i}), \nabla_{x_i} f^A_i(x_i, x_{-i}) \rangle - \frac{\beta_i}{2} \Delta t^2 ||\nabla_{x_i} f_i^A(x_i, x_{-i})||^2 \\
    &\ge f^A_i(x'_i,x'_{-i}) + \underbrace{||\nabla_{x_i} f_i^A(x_i, x_{-i})||^2 \Delta t - \xi_i \Delta t^2 - \gamma_i \Delta t^3}_{a_i}
\end{align}
where $\xi_i = \delta_i + \frac{\beta_i}{2} ||\nabla_{x_i} f_i^A(x_i, x_{-i})||^2$. The parameters $\xi_i$ and $\gamma_i$ are bounded, therefore, there exists a $\Delta t > 0$ small enough such that $a_i \ge 0$.
\end{proof}

\begin{theorem}[Local Smoothness]
\label{local_smooth}
Given $n$ losses, $f_i^A(\boldsymbol{x})$, $i\in \{1,\ldots,n\}$, with $\beta_i$-Lipschitz gradients there exists a $\Delta t > 0$ sufficiently small such that the game defined by these losses is smooth only if
\begin{align}
    \sum_{i=1}^n a_i &\le \lambda \sum_{i=1}^n f^A_i(x_i,x_{-i}) + (\mu-1) \sum_{i=1}^n f^A_i(x'_i,x'_{-i}) \,\, \forall x_i
\end{align}
where $x_i' = x_i - \Delta t \nabla_{x_i} f^A_i(\boldsymbol{x})$ and $a_i = ||\nabla_{x_i} f_i^A(x_i, x_{-i})||^2 \Delta t - \xi_i \Delta t^2 - \gamma_i \Delta t^3 \ge 0$. Note this is a necessary, not sufficient condition for a game to be globally smooth.
\end{theorem}
\begin{proof}
Plugging Lemma~\ref{taylor} into the original definition of smoothness for $x_i' = x_i - \Delta t \nabla_{x_i} f^A_i(\boldsymbol{x})$ and $\Delta t$ sufficiently small gives
\begin{align}
    \sum_{i=1}^n f^A_i(x'_i,x'_{-i}) + a_i \le \sum_{i=1}^n f^A_i(x_i,x'_{-i}) &\le \lambda \sum_{i=1}^n f^A_i(x_i,x_{-i}) + \mu \sum_{i=1}^n f^A_i(x'_i,x'_{-i}).
\end{align}
Rearranging the outer terms of the inequalities gives
\begin{align}
    \sum_{i=1}^n a_i &\le \lambda \sum_{i=1}^n f^A_i(x_i,x_{-i}) + (\mu-1) \sum_{i=1}^n f^A_i(x'_i,x'_{-i}).
\end{align}
\end{proof}

Note this is different than the definition of local smoothness in~\citep{roughgarden2015local}. 

\begin{theorem}
Given $n$ losses, $f_i^A(\boldsymbol{x})$, $i\in \{1,\ldots,n\}$, with $\beta_i$-Lipschitz gradients there exists a $\Delta t > 0$ sufficiently small such that the \textbf{utilitarian} local price of anarchy of the game (to $\mathcal{O}(\Delta t^2)$) is upper bounded by
\begin{equation}
    \rho \le \max_i\{1 + \Delta t \, \texttt{ReLU} \Big( \frac{d}{dt} \log(f_i^A(\boldsymbol{x})) + \frac{||\nabla_{x_i} f_i^A(\boldsymbol{x})||^2}{f_i^A(\boldsymbol{x})\bar{\mu}} \Big)\} \label{poa_theorem_proof}
\end{equation}
where $i$ indexes each agent and $\bar{\mu}$ is a user defined nonnegative scalar.
\end{theorem}

\begin{proof}
To ease exposition, let $b_i = f_i^A(x_i, x_{-i})$ and $c_i = f_i^A(x'_i, x'_{-i})$ so that local smoothness becomes
\begin{align}
    \sum_{i=1}^n a_i &\le \lambda \sum_{i=1}^n b_i + (\mu-1) \sum_{i=1}^n c_i.
\end{align}

If each agent $i$ ensures local \emph{individual} smoothness is satisfied, i.e.,
\begin{align}
\label{simple_smooth}
    a_i &\le \lambda_i b_i + (\mu_i-1) c_i,
\end{align}
then this is sufficient to satisfy local smoothness
\begin{align}
    \sum_{i=1}^n a_i &\le \max_i\{\lambda_i\} \sum_{i=1}^n b_i + (\max_i\{\mu_i\}-1) \sum_{i=1}^n c_i. \label{global_from_local}
\end{align}

Rearranging inequality~\ref{simple_smooth} and letting $\hat{\mu}_i = 1-\mu_i$, $\hat{a}_i = a_i/b_i$, and $\hat{c}_i = c_i/b_i$ gives
\begin{align}
    \lambda_i &\ge \frac{a_i}{b_i} - (\mu_i-1) \frac{c_i}{b_i} \\
    \lambda_i &\ge \hat{a}_i + \hat{\mu}_i \hat{c}_i.
\end{align}

Let each agent $i$ attempt to measure the local price of anarchy given the losses it observes on its trajectory and call this measure $\rho_i$. Then
\begin{align}
    \rho_i &= \inf_{\lambda_i, \hat{\mu}_i} \Big[ \frac{\lambda_i}{\hat{\mu}_i} \Big] \\
    &\text{ s.t. } \\
    \lambda_i &\ge \hat{a}_i + \hat{c}_i \hat{\mu}_i \label{con_smooth} \\
    \lambda_i &\ge \hat{\mu}_i \label{con_poa} \\
    \hat{\mu}_i &> 0 \label{con_feasible} \\
    \hat{\mu}_i &\le \bar{\mu} \label{con_bound}
\end{align}
where constraint~\ref{con_smooth} ensures local individual smoothness, constraint~\ref{con_poa} encodes that price of anarchy $\ge1$ by definition, and constraint~\ref{con_feasible} is required by the original conditions on $\mu$ for smoothness. Note that including an additional constraint for $\lambda_i > 0$ would be redundant and so is omitted. Constraint~\ref{con_bound} is optional and included to encode a prior by the agents on the smoothness parameters.

Recall that $\hat{a}_i$ and $\hat{c}_i$ are both non-negative; $\hat{c}_i$ controls the slope of constraint~\ref{con_smooth}. We can solve this optimization in closed form for the four distinct cases outlined in Figure~\ref{fig:derivation_figures}.
\begin{figure}[ht]
    \centering
    \includegraphics[scale=0.125]{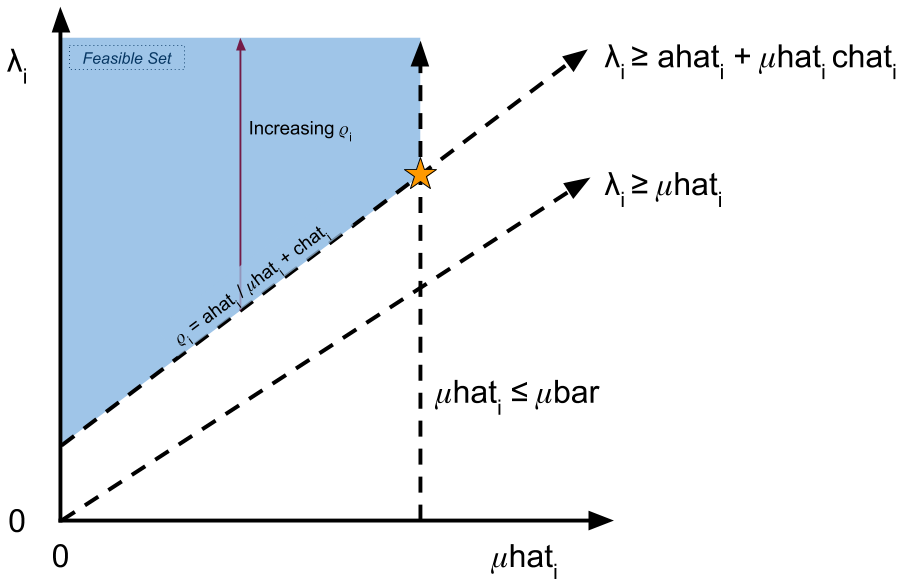}
    \includegraphics[scale=0.125]{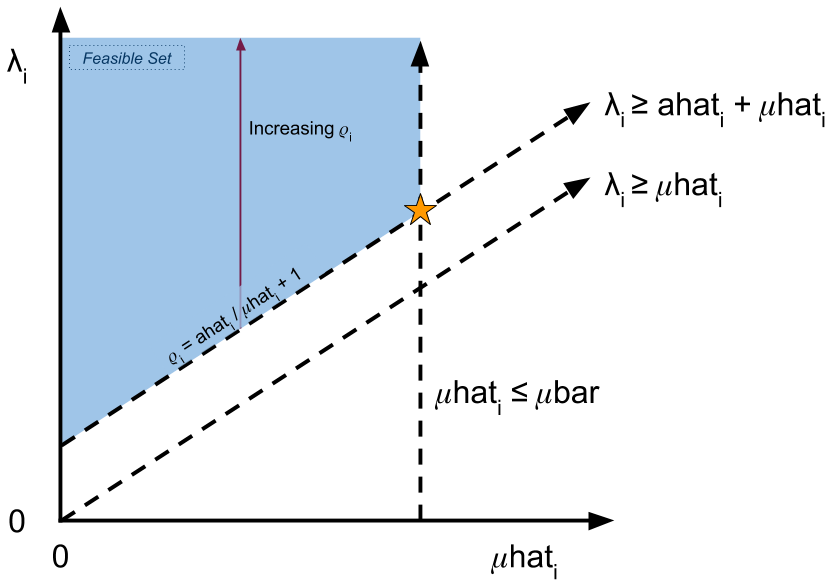}
    \includegraphics[scale=0.125]{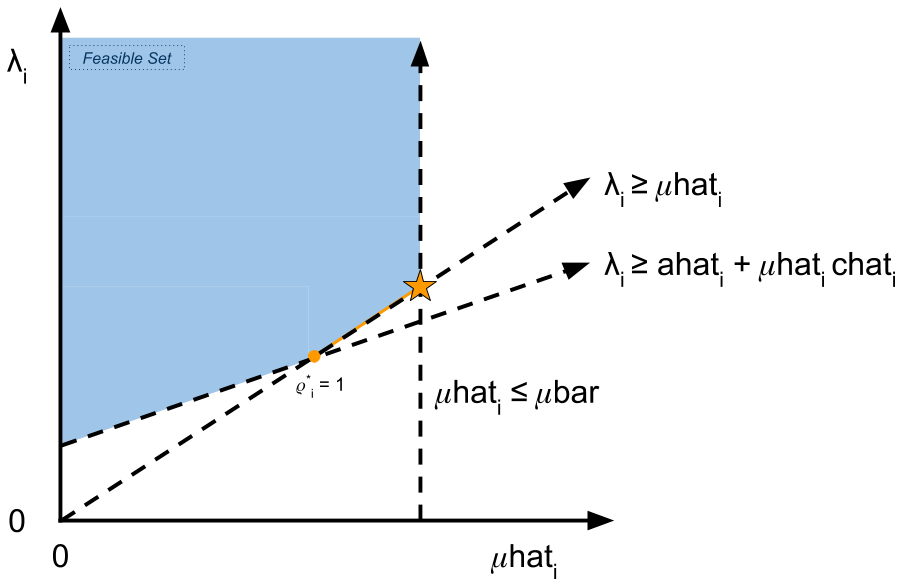}
    \includegraphics[scale=0.125]{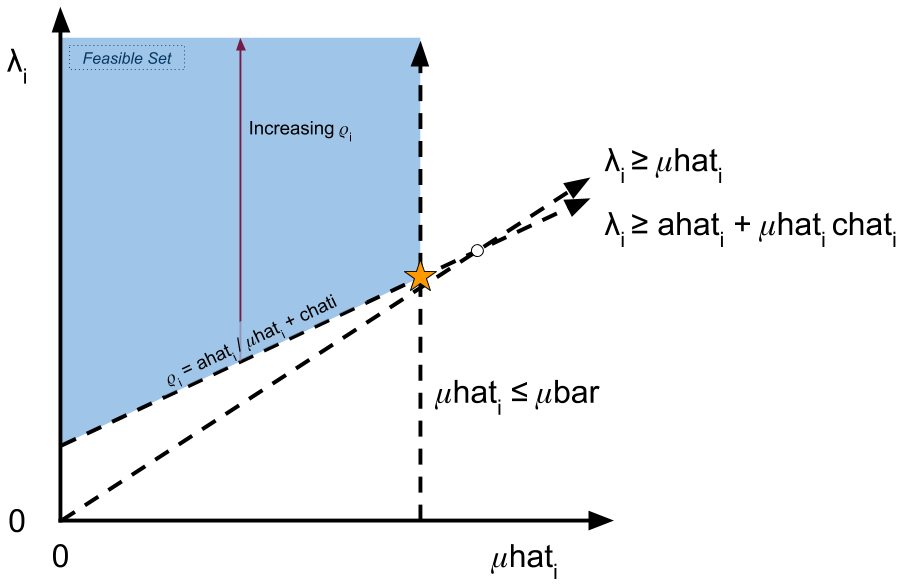}
    \caption{From left to right: a) $\hat{c}_i > 1$, b) $\hat{c}_i = 1$, c) $\hat{c}_i < 1$ and $\hat{c}_i + \frac{\hat{a}_i}{\bar{\mu}} \le 1$, d) $\hat{c}_i < 1$ and $\hat{c}_i + \frac{\hat{a}_i}{\bar{\mu}} > 1$.}
    \label{fig:derivation_figures}
\end{figure}

Figure~\ref{fig:derivation_figures} shows $\hat{\mu}$ always leads to minimal $\rho_i$ at $\bar{\mu}$, therefore $\max_i\{\mu_i\} = \max_i\{1-\hat{\mu}_i\} = 1-\bar{\mu}$. And so $\rho \le \frac{\max_i\{\lambda_i\}}{\bar{\mu}} = \max_i\{\rho_i\} = \max(1,\frac{\max_i\{\hat{a}_i + \bar{\mu} \hat{c}_i\}}{\bar{\mu}}) = \max(1,\max_i\{\frac{\hat{a}_i}{\bar{\mu}} + \hat{c}_i\})$. Assuming $\bar{\mu}$ is large allows us to approximate with $\max(1,\max_i\{\hat{c}_i\})$, so the local price of anarchy is determined by the largest increase in loss over all the agents; if all losses are decreasing, the local price of anarchy is $1$.

In summary, if $\hat{c}_i < 1$ and $\bar{\mu} \ge \frac{\hat{a}_i}{1-\hat{c}_i}$ (the intersection points of constraints~\ref{con_smooth} and~\ref{con_poa}), then $\rho_i = 1$. The latter inequality, $\frac{\hat{a}_i}{1-\hat{c}_i} \le \bar{\mu}$, can be rewritten as $\hat{c}_i \le 1 - \frac{\hat{a}_i}{\bar{\mu}}$. Alternatively, if $\hat{c}_i = 1$ and $\bar{\mu} \rightarrow \infty$ (i.e., constraint~\ref{con_bound} is omitted), $\rho_i$ also equals $1$. In all other cases, $\rho_i = \frac{\hat{a}_i}{\bar{\mu}_i} + \hat{c}_i$. If we assume $\hat{a}_i > 0$ (i.e., $||\nabla_{x_i} f_i^A(\boldsymbol{x})||>0$), we can reduce the cases above to
\begin{align}
\begin{cases}
    \rho_i = 1, & \text{ if } \hat{c}_i \le 1 - \frac{\hat{a}_i}{\bar{\mu}} \\
    \rho_i = \hat{c}_i + \frac{\hat{a}_i}{\bar{\mu}}, & \text{ else}.
\end{cases}
\end{align}

Let $\epsilon_i = \frac{\hat{a}_i}{\bar{\mu}} > 0$, then the two cases can be rewritten succinctly as
\begin{align}
    \rho_i &= \max(1, \hat{c}_i + \epsilon_i).
\end{align}

If we expand $\hat{c}_i$ as a series we find
\begin{align}
    \hat{c}_i &= \frac{f_i^A(\boldsymbol{x}')}{f_i^A(\boldsymbol{x})} \\
    &= \frac{f_i^A(\boldsymbol{x}) + \frac{d f^A_i(\boldsymbol{x})}{dt} \Delta t}{f_i^A(\boldsymbol{x})} + \mathcal{O}(\Delta t^2) \\
    &= 1 + \frac{\frac{d f^A_i(\boldsymbol{x})}{dt}}{f_i^A(\boldsymbol{x})} \Delta t + \mathcal{O}(\Delta t^2).
\end{align}

Therefore, to $\mathcal{O}(\Delta t^2)$,
\begin{align}
    \rho_i &= \max(1, 1 + \Big[ \frac{\frac{d f^A_i(\boldsymbol{x})}{dt}}{f_i^A(\boldsymbol{x})} + \overbrace{\frac{||\nabla_{x_i} f_i^A(\boldsymbol{x})||^2}{f_i^A(x_i, x_{-i})\bar{\mu}} \Big] \Delta t}^{\epsilon_i}) \\
    &= 1 + \Delta t \max(0, \frac{\frac{d f^A_i(\boldsymbol{x})}{dt}}{f_i^A(\boldsymbol{x})} + \frac{||\nabla_{x_i} f_i^A(\boldsymbol{x})||^2}{f_i^A(x_i, x_{-i})\bar{\mu}} ) \\
    &= 1 + \Delta t \, \texttt{ReLU} \Big( \frac{d}{dt} \log(f_i^A(\boldsymbol{x})) + \frac{||\nabla_{x_i} f_i^A(\boldsymbol{x})||^2}{f_i^A(x_i, x_{-i})\bar{\mu}} \Big) \\
    &= 1 + \Delta t \, \texttt{ReLU} \Big( \frac{d}{dt} \log(f_i^A(\boldsymbol{x})) \Big) \text{ as } \bar{\mu} \rightarrow \infty.
\end{align}
\end{proof}

The following lemma establishes that the proposed bound may be tight in some games although we do not conjecture that this bound is at all tight in general.
\begin{lemma}
The local $\rho$ bound with $\mu \rightarrow \infty$ in \eqref{poa_theorem_proof} is tight for some games.
\end{lemma}
\begin{proof}
Consider the two player game with loss functions $f_1(x_1) = x_1 - \kappa x_2$ and $f_2(x_2) = x_2 - \kappa x_1$ for players $1$ and $2$ respectively with $\kappa > 1$. Assume the player strategies are constrained to the line segment $x_1(\tau) = x_1 - \tau \Delta t$ and $x_2(\tau) = x_2 - \tau \Delta t$ with $\tau \in [0,1]$. Also, let $x_1 = x_2$ and recall each player is assumed to run gradient descent

Then $\frac{df_1}{dt} = \frac{\partial f_1}{\partial x1} \frac{dx_1}{dt} + \frac{\partial f_1}{\partial x2} \frac{dx_2}{dt} = \kappa - 1 > 0$. Similarly, $\frac{df_2}{dt} = \kappa - 1$. Given $x_1 = x_2$, the price of anarchy bound simplifies to $1 + \Delta t \texttt{ReLU} \frac{d}{dt} \log(f_i^A(\boldsymbol{x})) = 1 + \Delta t \texttt{ReLU} \frac{d/dt f_i^A(\boldsymbol{x})}{f_i^A(\boldsymbol{x})} = 1 + \Delta t \frac{\kappa-1}{f_i(\boldsymbol{x})}$.

Also, $f_1(x(\tau)) = x_1 - \tau \Delta t - \kappa (x_2 - \tau \Delta t) = x_1 - \kappa x_2 - \tau \Delta t (1 - \kappa) = f_1(\boldsymbol{x}) + \tau \Delta t (\kappa - 1)$. Likewise, $f_2(x(\tau)) = f_2(\boldsymbol{x}) + \tau \Delta t (\kappa - 1)$. By inspection, the Nash occurs where $x_1$ and $x_2$ are minimal along the segment at $\tau=1$, so $x_1^* = x_1 - \Delta t$ and $x_2^* = x_2 - \Delta t$. The values at Nash are $f_1(\boldsymbol{x}^*) = f_1(\boldsymbol{x}) + \Delta t (\kappa - 1)$ and $f_2(x(\tau)) = f_2(\boldsymbol{x}) + \Delta t (\kappa - 1)$. In contrast, optimal group loss, $\min_{x_1, x_2} (1 - \kappa)(x_1(\tau) + x_2(\tau))$, occurs at $\tau = 0$ and with values of $f_1(\boldsymbol{x})$ and $f_2(\boldsymbol{x})$. This implies the true price of anarchy is $1 + \Delta t \frac{2(\kappa-1)}{f_1(\boldsymbol{x}) + f_2(\boldsymbol{x})}$. Given $x_1 = x_2$, the true price of anarchy simplifies to $1 + \Delta t \frac{\kappa-1}{f_i(\boldsymbol{x})}$ which is the same as the upper bound.
\end{proof}
The goal of this work is to derive an approximate proxy that can be both easily estimated and optimized. The bound we derive relies on first order information. It would be interesting to tighten the bound with second order information or by computing the price of anarchy for an appropriate polymatrix approximation to the game.

\subsection{Accommodating Negative Loss Functions}
In experiments, we replace the second term, $\epsilon_i$, with a constant hyperparameter $\epsilon$:
\begin{align}
    \rho_i &= 1 + \Delta t \, \texttt{ReLU} \Big( \frac{d}{dt} \log(f_i^A(\boldsymbol{x})) + \epsilon \Big).
\end{align}

The $\log$ term appears due to price of anarchy being defined as the worst case Nash total loss divided by the minimal total loss. Although we have not defined an alternative price of anarchy, it is reasonable to believe one which defines the price of anarchy additively might drop the $\log$ term, leading to minimizing the following:
\begin{align}
    \hat{c}_i &= f_i^A(\boldsymbol{x}') - f_i^A(\boldsymbol{x}) \\
    &= f_i^A(\boldsymbol{x}) + \frac{d f^A_i(\boldsymbol{x})}{dt} \Delta t - f_i^A(\boldsymbol{x}) + \mathcal{O}(\Delta t^2) \\
    &= \frac{d f^A_i(\boldsymbol{x})}{dt} \Delta t + \mathcal{O}(\Delta t^2)
\end{align}
so that
\begin{align}
    \rho_i &= \Delta t \, \texttt{ReLU} \Big( \frac{d}{dt} f_i^A(\boldsymbol{x}) + \tilde{\epsilon} \Big)
\end{align}
where $\tilde{\epsilon}_i \approx \frac{||\nabla_{x_i} f_i^A(\boldsymbol{x})||^2}{\bar{\mu}}$ is replaced in experiments with a constant hyperparameter, $\tilde{\epsilon}$ as before. This objective is appealing as it does not require losses to be positive.

\subsubsection{Multiplicative vs Additive Price of Anarchy}
\label{multiplicative_poa}
In~\S\ref{decentralized}, we proposed an alternative gradient direction to the one derived in \eqref{lin_poa_dist}. This was a pragmatic change to make D3C amenable to games with negative loss, but may have appeared theoretically unappealing to the reader. Here, we show that the price of anarchy, as a multiplicative ratio, is already a somewhat arbitrary and non-robust choice.

Specifically, the price of anarchy of a game is not invariant to a global offset to the loss functions. Let the original price of anarchy of a game be $\frac{a}{b}$. Consider adding a constant $c$ to each of the $n$ losses in the game; note this does not change the locations of the Nash equilibrium or the total loss minimizer. However, the new price of anarchy becomes $\frac{a+nc}{b+nc} \rightarrow 1$ as $c \rightarrow \infty$. On the other hand, let $c \rightarrow -b/n$ from the right. Then the new price of anarchy approaches infinity. In summary, the price of anarchy, as defined multiplicatively, can be made arbitrarily large or small by adding a constant to each loss function in the game.

By removing the $\log$ term from the gradient, $\nabla_{A_i} \rho_i$, we effectively removed this effect. Lastly, the most important and general part of gradient direction, $\nabla_{A_i} \rho_i$, is the the Improve-Stay, Suffer-Shift component which is retained in $\tilde{\nabla}_{A_i} \rho_i$.

\subsubsection{Why Minimize $\frac{d}{dt} f_i^A(\boldsymbol{x})$ w.r.t. $A_i$? Why Not $\frac{d}{dt} f_i(\boldsymbol{x})$?}
The local price of anarchy is defined using the time derivative of the transformed loss. Instead, can agents minimize the time derivative of their original loss w.r.t. $A_i$? Note the dependence on $A_i$ appears in the time derivative terms through the update dynamics, e.g. $\frac{dx_i}{dt} = \frac{dx_i}{dt}(A)$.

In our loss mixing model, agent $i$ can influence the update of agent $j$ directly through $A_i$. This occurs because the transformed losses are computed using $A^\top$ and so $A_{ij}$ is used to re-mix agent $j$'s loss. This allows agent $i$ to affect the $h(\frac{dx_{j \ne i}}{dt})$ terms mentioned back in~\S\ref{gdwgd} and~\S\ref{game_is_prob}, circumventing the issues originally discussed in those sections.


However, we conducted experiments on the prisoner's dilemma using this approach, and although minimizing $\frac{d}{dt} f_i(\boldsymbol{x})$ w.r.t. $A_i$ worked for the $2$-player variant, it failed to minimize the price of anarchy for $3$, $5$, or $10$ players. Therefore, we discontinued its use in further experiments.

\subsection{Egalitarian Price of Anarchy}
\label{egalitarian}
If the objective of interest is \emph{egalitarian} rather than \emph{utilitarian}, then a game is $(\lambda, \mu)$-smooth instead if:
\begin{align}
    \sum_{i=1}^n f^A_i(x_i,x'_{-i}) &\le \lambda \max_{i} f^A_i(x_i,x_{-i}) + \mu \max_{i} f^A_i(x'_i,x'_{-i})
\end{align}
for all $\boldsymbol{x}, \boldsymbol{x}' \in \mathcal{X}$ where $\lambda > 0$, $\mu < 1$, and $\max_i f_i^A(\boldsymbol{x})$ is assumed to be non-negative for any $\boldsymbol{x} \in \mathcal{X}$.

The price of anarchy, $\rho_e$, gives the ratio of the worst case Nash max-loss to the minimal max-loss:
\begin{align}
    \rho &= \frac{\max_{\mathcal{X}^*} \max_i f^A_i(\boldsymbol{x}^*)}{\min_{\mathcal{X}} \max_i f^A_i(\boldsymbol{x})} \ge 1 \\
    &\le \inf_{\lambda>0, \mu<1} \Big[ \frac{\lambda}{1-\mu} \Big]
\end{align}
where $\boldsymbol{x}^*$ is an element of the set of Nash equilibria, $\mathcal{X}^*$.

\begin{theorem}
Given $n$ losses, $f_i^A(\boldsymbol{x})$, $i\in \{1,\ldots,n\}$, with $\beta_i$-Lipschitz gradients there exists a $\Delta t > 0$ sufficiently small such that the local \textbf{egalitarian} price of anarchy of the game (to $\mathcal{O}(\Delta t^2)$) is upper bounded by
\begin{equation}
    \rho_e \le 1 + \Delta t \, \texttt{ReLU} \Big( \frac{d}{dt} \log(\max_i\{f_i^A(\boldsymbol{x})\}) + \frac{\sum_{i=1}^n ||\nabla_{x_i} f_i^A(\boldsymbol{x})||^2}{\bar{\mu} \max_i f_i^A(\boldsymbol{x})} \Big).
\end{equation}
where $i$ indexes each agent and $\bar{\mu}$ is a user defined nonnegative scalar.
\end{theorem}

\begin{proof}
By Lemma~\ref{taylor},
\begin{align}
    \sum_{i=1}^n f^A_i(x'_i,x'_{-i}) + a_i \le \sum_{i=1}^n f^A_i(x_i,x'_{-i}) &\le \lambda \max_{i} f^A_i(x_i,x_{-i}) + \mu \max_{i} f^A_i(x'_i,x'_{-i}).
\end{align}
Rearranging the outer terms of the inequalities gives
\begin{align}
    \sum_{i=1}^n a_i &\le \lambda \max_{i} f^A_i(x_i,x_{-i}) + \mu \max_{i} f^A_i(x'_i,x'_{-i}) - \sum_{i=1}^n f^A_i(x'_i,x'_{-i}) \\
    &\le \lambda \max_{i} f^A_i(x_i,x_{-i}) + (\mu-1) \max_{i} f^A_i(x'_i,x'_{-i}) \\
    \implies a &\le \lambda b + (\mu-1) c.
\end{align}
where $a = \sum_{i=1}^n a_i$, $b = \max_{i} f^A_i(x_i,x_{-i})$, and $c = \max_{i} f^A_i(x'_i,x'_{-i})$. The proof proceeds as before in the \emph{utilitarian} case except the price of anarchy does not decompose into a max over agent-centric estimates.
\end{proof}

\section{Description of Games in Experiments}
We describe the traffic network and prisoner's dilemma games in detail here. We point the reader to~\citep{eccles2019imitation} for further details of Coins and~\citep{hughes2018inequity} for Cleanup.

\subsection{Generating Networks that Exhibit Braess's Paradox}
In order to randomly generate a traffic network exhibiting Braess's paradox, it is sufficient to guarantee two properties. One is that the shortcut route is a strictly dominant path (shorter commute time). This ensures all agents take the shortcut in the Nash equilibrium. The other is that there exists a joint strategy avoiding the shortcut with lower total commute time than all agents taking the shortcut. We assume there are four drivers.

\begin{figure}[ht]
    \centering
    \includegraphics[scale=0.5]{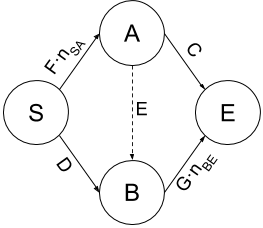}
    \caption{A theoretical traffic nework with congestion parameters, $F$ and $G$, and constant commute time parameters $C$, $D$, and $E$.}
    \label{fig:rand_braess_net}
\end{figure}

The shortcut, SABE, is a strictly dominant (strictly shorter commute) if

\begin{align}
    Fn_{sa} + Gn_{be} + E &< Fn_{sa} + C \\
    Fn_{sa} + Gn_{be} + E &< Gn_{be} + D \\
    \implies E &< \min\{C - Gn_{be}, D - Fn_{sa}\} \\
    \implies G &< \frac{C}{n_{be}} \text{ which is ensured if } C > 4G \\
    \implies F &< \frac{C}{n_{sa}} \text{ which is ensured if } D > 4F.
\end{align}

And there exists a pure joint strategy with at least $\Delta$ less total commute time if
\begin{align}
    \tau_{Nash} &= 4(4(F+G) + E) \\
    \tau_{Opt} &= \argmin_{n_{sa} \in \{1, 2, 3\}, n_{be} = 4 - n_{sa}} \{ n_{sa} (F n_{sa} + C) + n_{be} ( G n_{be} + D)\} \\
    \tau_{Nash} > \tau_{Opt} + \Delta &\implies E > \frac{\tau_{opt} + \Delta}{4} - 4(F+G).
\end{align}

So we can randomly generate a Braess network with Algorithm~\ref{gen_braess}.
\begin{algorithm}[ht]
\begin{algorithmic}
    \STATE \texttt{fail} $\leftarrow$ True
    \WHILE{\texttt{fail}}
        \STATE $F \sim \{1, \ldots, 20\}$
        \STATE $G \sim \{1, \ldots, 20\}$
        \STATE $C \sim \{4G + 10, \ldots, 4G + 20\}$$\,\,\triangleright\,$$10$ is an arbitrary buffer
        \STATE $D \sim \{4F + 10, \ldots, 4F + 20\}$$\,\,\triangleright\,$$20$ is an arbitrary upper limit
        \STATE $\tau_{Opt} \leftarrow \argmin_{n_{sa} \in \{1, 2, 3\}, n_{be} = 4 - n_{sa}} \{ n_{sa} (F n_{sa} + C) + n_{be} ( G n_{be} + D)\}$
        \STATE $E_{\min} = \max\{\frac{\tau_{Opt} + \Delta}{4} - 4 (F + G), 0\}$
        \STATE $E_{\max} = \min\{C - 4G, D - 4F\}$
        \IF{$E_{\min} < E_{\max}$}
            \STATE \texttt{fail} $\leftarrow$ False
            \STATE $E \sim \{E_{\min}, \ldots, E_{\max}\}$
        \ENDIF
    \ENDWHILE
    \STATE Output: $C$, $D$, $E$, $F$, $G$
\end{algorithmic}
\caption{\texttt{gen\_braess}}
\label{gen_braess}
\end{algorithm}


The expected commute times for this Braess network can be computed exactly given stochastic commuting policies. Consider a network with four drivers and let $x_{ij}$ specify the probability of driver $i$ taking route $j$ through the network. Then let
\begin{align}
    \boldsymbol{x} &= \begin{bmatrix}
        x_{11} \\
        x_{12} \\
        x_{13} \\
        \vdots \\
        x_{41} \\
        x_{42} \\
        x_{43}
    \end{bmatrix}, \, \boldsymbol{b} = \begin{bmatrix}
        C \\
        D \\
        E
    \end{bmatrix}, \, M = \begin{bmatrix}
        F & 0 & F \\
        0 & G & G \\
        F & G & F + G
    \end{bmatrix} \succeq 0 , \, \boldsymbol{b}_r = \begin{bmatrix}
        \boldsymbol{b} \\
        \boldsymbol{b} \\
        \boldsymbol{b} \\
        \boldsymbol{b}
    \end{bmatrix}, \, M_r = \begin{bmatrix}
        M \\
        M \\
        M \\
        M
    \end{bmatrix}, \, I = \begin{bmatrix}
        1 & 0 & 0 \\
        0 & 1 & 0 \\
        0 & 0 & 1
    \end{bmatrix}
\end{align}
and let
\begin{align}
    S &= \begin{bmatrix}
        \boldsymbol{I} & \boldsymbol{I} & \boldsymbol{I} & \boldsymbol{I}
    \end{bmatrix}, \, A_i = \begin{bmatrix}
        \mathds{1}(i==1) \boldsymbol{I} & \mathbf{0} & \mathbf{0} & \mathbf{0} \\
        \mathbf{0} & \mathds{1}(i==2) \boldsymbol{I} & \mathbf{0} & \mathbf{0} \\
        \mathbf{0} & \mathbf{0} & \mathds{1}(i==3) \boldsymbol{I} & \mathbf{0} \\
        \mathbf{0} & \mathbf{0} & \mathbf{0} & \mathds{1}(i==4) \boldsymbol{I}
    \end{bmatrix}.
\end{align}
Then $\boldsymbol{\tau}_r = M_r S \boldsymbol{x} + \boldsymbol{b}_r$ gives commute time for each path replicated for four agents:
\begin{align}
    \boldsymbol{\tau}_r &= M_r S \boldsymbol{x} + \boldsymbol{b}_r \\
    &= \begin{bmatrix}
        \text{top route time for player $1$} \\
        \text{bottom route time for player $1$} \\
        \text{shortcut time for player $1$} \\
        \vdots \\
        \text{top route time for player $4$} \\
        \text{bottom route time for player $4$} \\
        \text{shortcut time for player $4$}
    \end{bmatrix}.
\end{align}
The expected commute time for agent $1$ is just the inner product of the first $3$ entries of this vector with agent $1$'s policy. We use the matrix $A_i$ to effectively select the appropriate commute times from $\tau_r$. Continuing, let
\begin{align}
    Q_i &= A_i^\top M_r S \\
    d_i &= A_i^\top \boldsymbol{b}_r = A_i \boldsymbol{b}_r \\
    C_i &= Cov(x_i) = \texttt{diag}(x_i) - x_i x_i^\top \\
    C &= Cov(\boldsymbol{x}) = \texttt{block\_diag}(C_i).
\end{align}
We can now write agent $i$'s loss as
\begin{align}
    l_i(\boldsymbol{x}) &= (A_i \boldsymbol{x})^\top \boldsymbol{\tau}_r \\
    &= \boldsymbol{x}^\top Q_i \boldsymbol{x} + d_i^\top \boldsymbol{x} \\
    \mathbb{E}[l_i(\boldsymbol{x})] &= \mathbb{E}[\boldsymbol{x}^\top Q_i \boldsymbol{x}] + d_i^\top \boldsymbol{x} \\
    &= \trace(Q_i C) + \boldsymbol{x}^\top Q_i \boldsymbol{x} + d_i^\top \boldsymbol{x} \\
    &= \trace(M C_i) + \boldsymbol{x}^\top Q_i \boldsymbol{x} + d_i^\top \boldsymbol{x}
\end{align}
which is easily amenable to analysis and makes the fact that the loss is quadratic, readily apparent.

\subsection{A Reformulation of the Prisoner's Dilemma}
\label{pd_convex}
In an $n$-player prisoner's dilemma, each player must decide to defect or cooperate with each of the other players creating a combinatorial action space of size $2^{n-1}$. This requires a payoff tensor with $2^{n(n-1)}$ entries. Instead of generalizing prisoner's dilemma~\citep{rapoport1965prisoner} to $n$ players using $n$th order tensors, we translate it to a game with convex loss functions.
Figure~\ref{fig:pd_table_example_app} shows how we can accomplish this.
\begin{figure}[ht]
    \centering
    \includegraphics[scale=0.5]{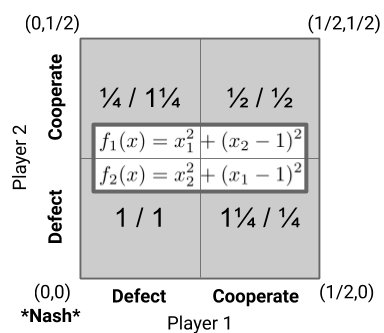}
    \caption{A reformulation of the prisoner's dilemma using convex loss functions instead of a normal form payoff table.}
    \label{fig:pd_table_example_app}
\end{figure}
Generalizing this to $n$ players, we say that for all $i, j, k$ distinct, 1) player $i$ wants to defect against player $j$, 2) player $i$ wants player $j$ to defect against player $k$, and 3) player $i$ wants player $j$ to cooperate with itself. In other words, each player desires a free-for-all with the exception that no one attacks it. See~\S\ref{pd_convex} for more details.

For example, we can define the vector of loss functions succinctly for three players with
\begin{align}
    \boldsymbol{f}(\boldsymbol{x}) &= \sum_{columns} \Big[ \Big( \begin{bmatrix}
    \boldsymbol{x}^\top \\ \boldsymbol{x}^\top \\ \boldsymbol{x}^\top
    \end{bmatrix}
    - C \Big)^2 \Big]
\end{align}
where $\boldsymbol{x} = [x_{ij}]$ is a column vector ($i \in [1,n], j \in [1,n-1]$, construct $\boldsymbol{x}$ as a matrix and then flatten in major-row order) containing the player strategies, $C$ is an $n \times n(n-1)$ matrix with entries that either equal $0$ or $c \in \mathbb{R}^+$, and the exponentation, $(\cdot)^2$, is performed elementwise.

More specifically, $C$ is a circulant matrix with column order reversed. For example, the matrix $C$ associated with the three player game is
\begin{align}
    C &= \begin{bmatrix}
    0 & 0 & c & 0 & 0 & c \\
    0 & c & 0 & 0 & c & 0 \\
    c & 0 & 0 & c & 0 & 0
    \end{bmatrix}
\end{align}
where $c>0$. Setting $x_{ij}=0$ encodes that player $i$ has defected against its $j$th opponent. In the first row of $C$ above, the first two entries can be read as player $1$ is incentivized to defect against players $2$ and $3$. The next two entries state that player $1$ receives a penalty if player $2$ doesn't cooperate, but wants player $2$ to defect against player $3$. The final two entries state that player $1$ receives a penalty if player $3$ doesn't cooperate, but wants player $3$ to defect against player $2$.
The matrix, $C$, can be constructed for $n$-player games with numpy~\citep{oliphant2006guide} as
\begin{verbatim}
    row = numpy.array(([0]*(n-1)+[c])*(n-1))[::-1]
    C = scipy.linalg.circulant(row1)[:n,::-1]
\end{verbatim}
Note that this matrix is of size $n \times n(n-1)$ containing $\mathcal{O}(n^3)$ entries.

The minimal total loss for this problem is $(n-1)^2 c^2$ and occurs at $x_{ij}=\frac{c}{n}$:
\begin{align}
    f_{\text{total}} &= \mathbf{1}^\top \vec{f}(\boldsymbol{x}) = \sum_{i=1}^n \sum_{j=1}^{n-1} (n-1)x_{ij}^2 + (x_{ij}-c)^2 \\
    \frac{\partial f_{\text{total}}}{\partial x_{ij}} &= 2(n-1)x_{ij} + 2(x_{ij}-c) = 0 \\
    \implies x_{ij} &= \frac{c}{n} \\
    \implies f_{\text{total}} &= n(n-1)\Big[ \frac{(n-1)c^2}{n^2} + \frac{(n-1)^2c^2}{n^2}\Big] = (n-1)^2 c^2. \label{opt_welfare}
\end{align}
Nash occurs at the origin. This can be quickly derived by leveraging variational inequality theory~\citep{facchinei2007finite,nagurney2012projected} and noticing that the Jacobian of gradient descent dynamics is $2 \boldsymbol{I}$, hence strongly monotone. Strongly monotone variational inequalities have unique a Nash equilibrium coinciding with the strategy set at which the gradients are all zero (assuming this point lies in $\mathcal{X}$). The total loss at Nash ($x_{ij}=0$) is $n(n-1)c$ by inspection.

\subsubsection{Cooperation Robust to Mavericks}
\label{d3c_robust_in_pd}
\begin{proposition}
    In heterogeneous populations containing both D3C agents and selfish (gradient descent) agents, D3C agents end up with strictly lower loss when playing the proposed reformulation of the prisoner's dilemma.
\end{proposition}
\begin{proof}

Note that player $i$ controls variables $x_{ij}$ and suffers loss $f_i(\boldsymbol{x})$. Assume some subset of the players defect and play some fixed strategy. Let this subset be the players $1$ through $m$ w.l.o.g. because the player losses are symmetric. The remaining player (non-defector) losses can be rewritten as
\begin{align}
    \boldsymbol{f}_{i>m}(\boldsymbol{x}) &= \boldsymbol{f}(\boldsymbol{x} \vert C_{\{i>m,j>m(n-1)\}}) + \mathcal{K}
\end{align}
where $\mathcal{K}$ is some vector-valued constant independent of these non-defectors' strategies. Due to the structure of $C$, the losses that remain simply represent a ($n-m$)-player prisoner's dilemma. To see this, consider player 1 defecting in a 3-player prisoner's dilemma, i.e., consider the $C_{\{i>1,j>2\}}$ submatrix. The loss functions for players 2 and 3 depend in exactly the same way on the variables $x_{21}$ and $x_{32}$, i.e., $(x_{21}-0)^2 + (x_{32}-0)^2 + \cdots$, therefore, they will both agree on setting $x_{21}=x_{32}=0$. The game that remains is exactly the 2-player prisoner's dilemma between players 2 and 3. So assuming these players run our proposed algorithm (D3C), they will converge to minimizing total loss of this subgame.

Of particular interest is the case where the defectors naively play fixed selfish strategies, i.e., $x_{ij}=0$. In this case, cooperating agents not only achieve lower subgroup loss, but also lower individual loss.

Recall that the loss for each player when all defect (naive selfish play implies $x_{ij}=0$) is $n-1$. If only a subset of players defect and the remaining cooperate, the defectors achieve losses greater than $n-1$\textemdash this can be seen from the fact that $x_{ij}=0$ is a strict Nash. Therefore, if we show that a cooperator's loss is less than $n-1$, we prove that cooperators outperform defectors.


Each defector adds 1 to the loss of a cooperator and the loss due to the cooperators' subgame prisoner's dilemma is $\frac{(n-m-1)^2}{n-m}$ (\eqref{opt_welfare}). Therefore, the loss of a cooperator is $m + \frac{(n-m-1)^2}{n-m}$.
The loss of a defector is always greater:
\begin{align}
    \underbrace{(n-1)}_{\text{defector}} - \underbrace{m - \frac{(n-m-1)^2}{n-m}}_{\text{cooperator}} &= (n-m-1) - \frac{(n-m-1)^2}{n-m} \\
    &= \frac{(n-m)(n-m-1)-(n-m-1)^2}{n-m} = \frac{n-m-1}{n-m} > 0.
\end{align}
\end{proof}

\section{Additional Experiments}
We present additional results on four RL experiments, one small game as another counterargument to welfare-maximization, and a negative result for local $\rho$-minimization (which D3C is an instance of).

\subsection{Prisoner's Dilemma}
\label{appx:more_pd}
Figures~\ref{fig:pd_results_statistics_2player} and~\ref{fig:pd_results_statistics_10player} further support that D3C with a randomly initialized strategy successfully minimizes the price of anarchy. In contrast, gradient descent learners provably converge to Nash at the origin with $\rho=\frac{n}{c(n-1)}$. The price of anarchy grows unbounded as $c \rightarrow 0$.
\begin{figure}[ht!]
    \centering
    \includegraphics[width=0.23\textwidth]{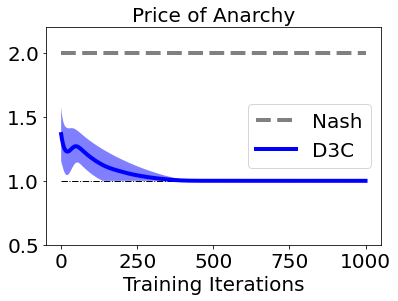}
    \includegraphics[width=0.23\textwidth]{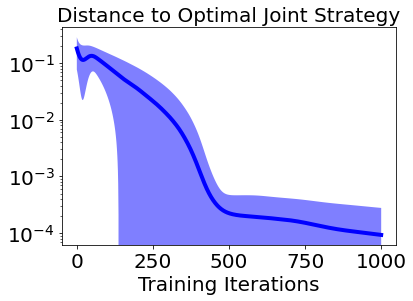}
    \caption{\textbf{Prisoner's Dilemma ($n=2, c=1, \rho=2$)}\textemdash Convergence to $\rho=1$ (left) and the unique optimal joint strategy (right) over $1000$ runs. The shaded region captures $\pm$ $1$ standard deviation around the mean (too small to see on left). Gradient descent (not shown) provably converges to Nash.}
    \label{fig:pd_results_statistics_2player}
\end{figure}
\begin{figure}[ht!]
    \centering
    \includegraphics[width=0.23\textwidth]{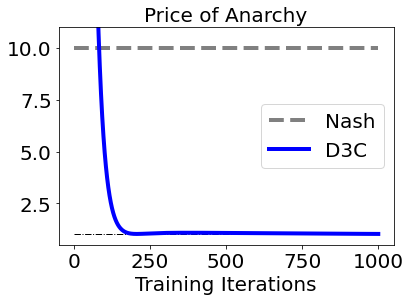}
    \includegraphics[width=0.23\textwidth]{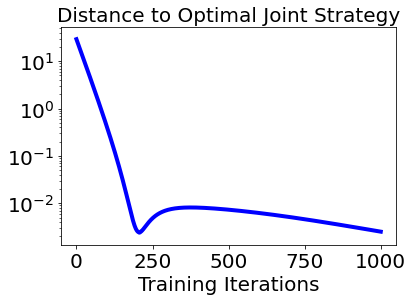}
    \caption{\textbf{Prisoner's Dilemma ($n=10, c=\frac{1}{9}, \rho=10$)}\textemdash Convergence to $\rho=1$ (left) and the unique optimal joint strategy (right) over $1000$ runs. The shaded region captures $\pm$ $1$ standard deviation around the mean (too small to see on left). Gradient descent (not shown) provably converges to Nash.}
    \label{fig:pd_results_statistics_10player}
\end{figure}

\subsection{Trust-Your-Brother}
In this game, a predator chases two prey around a table. The predator is a bot with a hard-coded policy to move towards the nearest prey unless it is already adjacent to a prey, in which case it stays put. If the prey are equidistant to the predator, the predator flips a coin and moves according to the coin flip. The prey receive $0$ reward if they chose not to move and $-.01$ if they attempted to move. They additionally receive $-1$ if the predator is adjacent to them after moving.
\begin{figure}[ht!]
    \centering
    \includegraphics[scale=0.5]{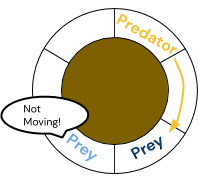}
    \hspace{1.0cm}
    \includegraphics[scale=0.5]{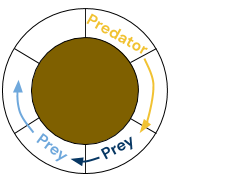}
    \caption{\textbf{Trust-Your-Brother} A bot chases agents around a table. The predator's prey can only escape if the other prey simultaneously moves out of the way. Selfish (left), cooperative (right).}
    \label{fig:trustyourbro_game_visual_app}
\end{figure}

The prey employ linear softmax policies (no bias term) and train via REINFORCE~\citep{williams1992simple}. Both prey receive the same $2$-d observation vector. The first feature specifies the counter-clockwise distance to the predator minus the clockwise distance for the dark blue prey. The second feature specifies the same for the light blue prey. Episodes last $5$ steps and there are $6$ grid cells in the ring around the table as shown in Figure~\ref{fig:trustyourbro_game_visual_app}.

Figure~\ref{fig:trustyourbro_training} shows D3C approaches maximal total return over training; this is achieved by the agents compromising on their original reward incentives and paying more attention to those of the other agent during training as revealed by Figure~\ref{fig:trustyourbro_results}.
\begin{figure}[ht!]
    \centering
    \includegraphics[scale=0.4]{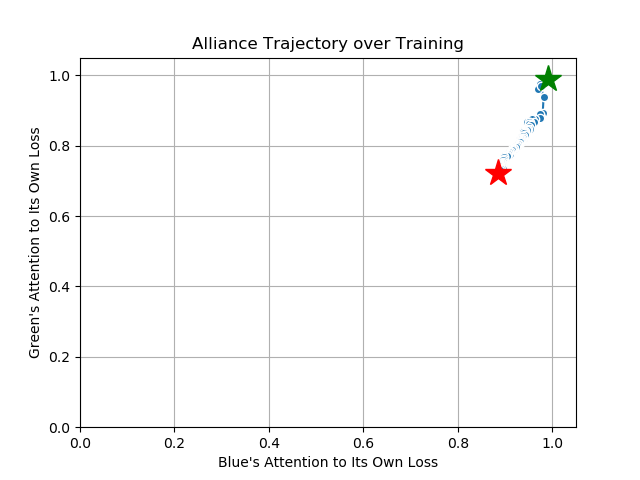}
    \caption{Agents are initialized to attend to their own losses. The trajectory here shows the agents compromising and adjusting to a mixture of losses (start at green, end at red star).}
    \label{fig:trustyourbro_results}
\end{figure}

\subsection{LIO Comparison}
\label{lio}
\citet{yang2020learning} propose an algorithm LIO (Learning to Incentivize Others) that equips agents with ``gifting'' policies represented as neural networks. At each time step, each agent observes the environment and actions of all other agents to determine how much reward to gift to the other agents. The parameters of these networks are adjusted to maximize the original environment reward (without gifts) minus some penalty regularizer for gifting meant to approximately maintain \emph{budget-balance}. In order to perform this maximization, each agent requires access to every other agents action-policy, gifting-policy, and return making this approach difficult to scale and decentralize.

\citet{yang2020learning} demonstrate LIO's ability to maximize welfare and achieve division of labor on a restricted version of the Cleanup game with high apple re-spawn rates and where agents are constrained to facing in one direction (compare Figure~3 of~\citep{yang2020learning} with Figure~1A of~\citep{hughes2018inequity}). While \citet{yang2020learning} show AC failing to achieve maximal welfare, we found the opposite result using A2C~\citep{espeholt2018impala} in Figure~\ref{fig:LIO}. In Figure~\ref{fig:LIO}, we also see that D3C is able to achieve near optimality. LIO appears to be approach maximal welfare as well in Figure~6C, therefore, this environment setting does not appear to differentiate the two approaches.
\begin{figure}[ht!]
    \centering
    \includegraphics[scale=0.5]{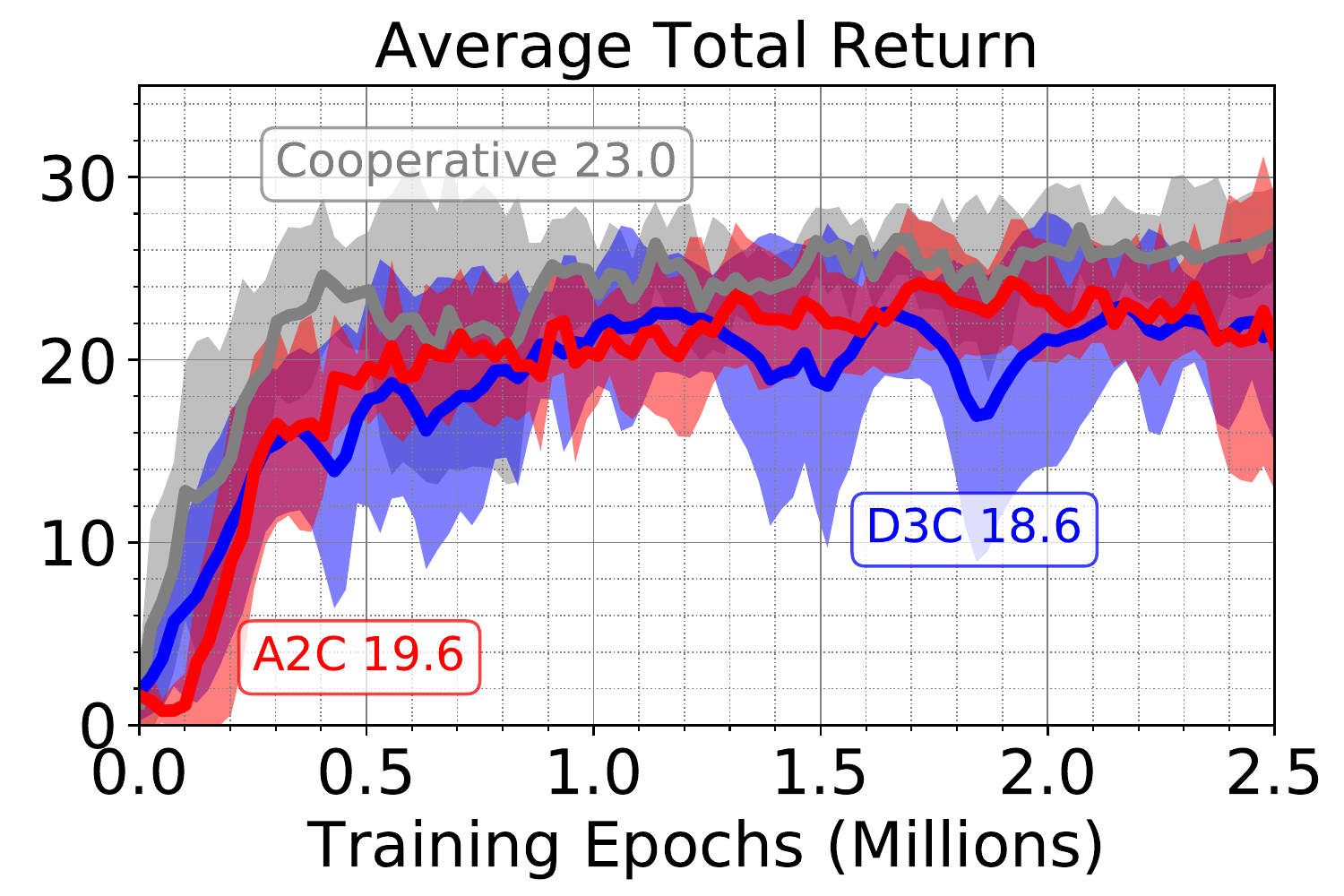}
    \caption{\textbf{Mini-Cleanup} Comparison against the mini Cleanup environment described in~\citep{yang2020learning}. In LIO, each agent requires access to every other agent's policy which makes implementing it within our decentralized codebase intractable. We suggest comparing the asymptotes of this plot with that of Figure~6C in~\citep{yang2020learning}.}
    \label{fig:LIO}
\end{figure}


\subsection{HarvestPatch}
\label{harvestpatch}
\citet{mckeesocial} introduce HarvestPatch as a common-pool resource game where apples spawn in predefined patches throughout a map. Agents must abstain from over-farming patches to the point of extinction by distributing their apple consumption as a group evenly across patches.

Figure~\ref{fig:harvestpatch} compares D3C against direct welfare maximization (Cooperation) and individual agent RL (A2C) on HarvestPatch.
\begin{figure}[ht!]
    \centering
    \includegraphics[scale=0.5]{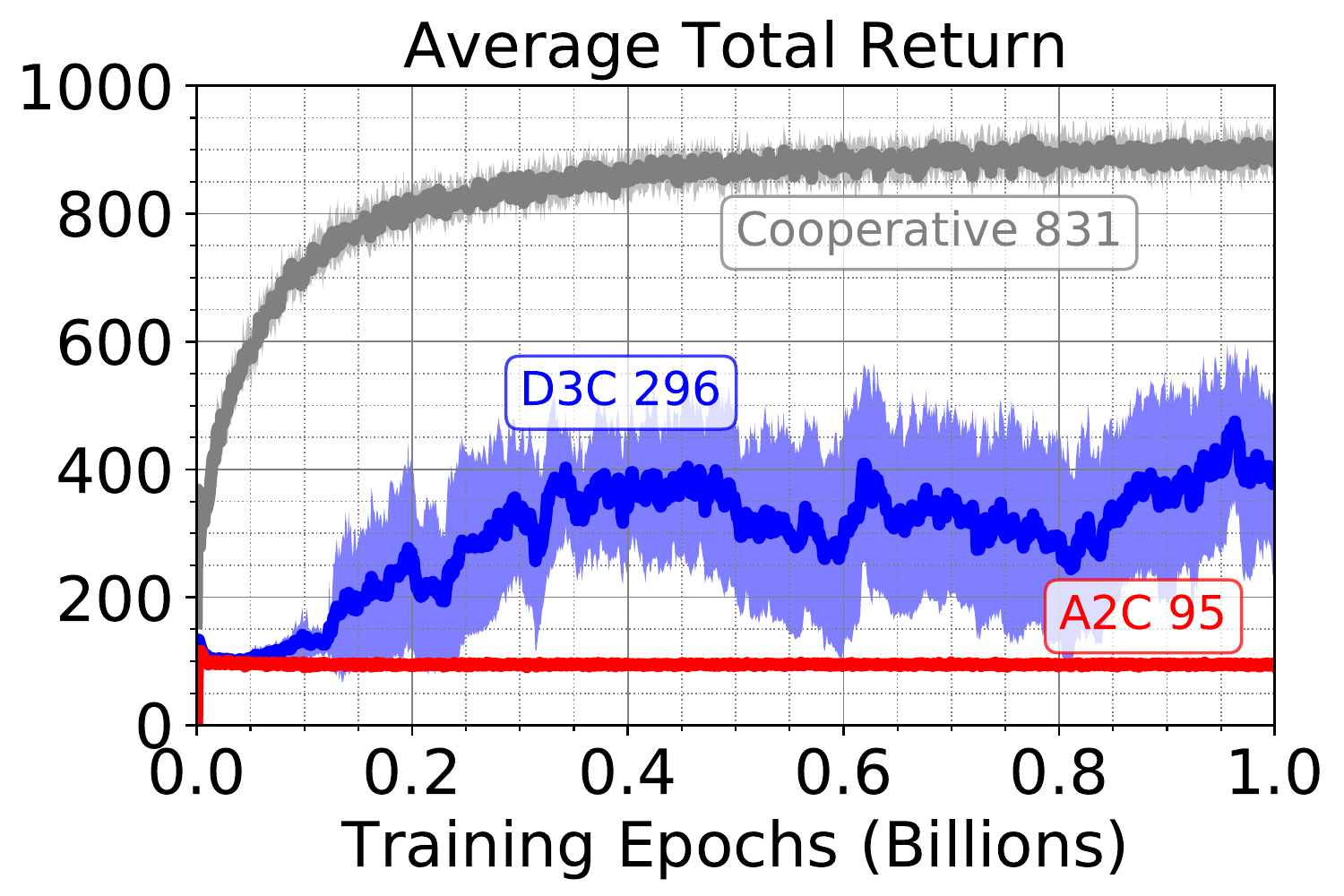}
    \caption{\textbf{HarvestPatch} Comparison against the HarvestPatch environment described in~\citep{mckeesocial}. D3C is able to increase welfare over the baseline approach of A2C at a slow rate.}
    \label{fig:harvestpatch}
\end{figure}

\subsection{A Zero-Sum Election}
\label{election}
Consider a hierarchical election in which two parties compete in a zero-sum game\textemdash for example, only one candidate becomes president. If, at the primary stage, candidates within one party engage in negative advertising, they hurt their chances of winning the presidential election because these ads are now public. This presents a prisoner's dilemma within each party. The goal then is for each party to solve their respective prisoner's dilemma and come together as one team, but certainly not maximize welfare\textemdash the zero-sum game between the two parties should be retained. A simple simulation with two parties consisting of two candidates each initially participating in negative advertising converges to the desired result after running D3C.
The final $4 \times 4$ loss mixing matrix, $A$, after training $1000$ steps is an approximate block matrix with $0.46$ on the $2 \times 2$ block diagonal and $0.04$ elsewhere:
\begin{align}
    A &= \begin{bmatrix}
        \mathbf{0.45795} & \mathbf{0.45708} & 0.04248 & 0.04248 \\
        \mathbf{0.45709} & \mathbf{0.45794} & 0.04248 & 0.04248 \\
        0.03778 & 0.03778 & \mathbf{0.46023} & \mathbf{0.46421} \\
        0.03778 & 0.03778 & \mathbf{0.46421} & \mathbf{0.46022}
    \end{bmatrix} \nonumber.
\end{align}
We make a \emph{duck-typing} argument that when multiple agents are optimizing the same loss, they are just components of a single agent because mathematically, there is no difference between this multiagent system and a single agent optimization problem. This matrix then indicates that two approximate teams have formed: the first two agents captured by the upper left block and vice versa.
Furthermore, the final eigenvalues of the game Jacobian are ($1.84 \pm 0.21i$) $\times 2$; perfect team formation gives ($2 \pm 0.25i$) $\times 2$. The existence of imaginary eigenvalues indicates that the zero-sum component of the game is retained. In contrast, minimizing total loss gives $0$ imaginary part because Hessians ($\texttt{Jac}(\nabla)$) are symmetric.

\subsection{Implicit Inequity Aversion}
\label{ineq_aver}

Welfare optimization can lead to poor outcomes as well, creating great inequity~\citep{bertsimas2011price,bertsimas2012efficiency,gemici2018wealth}. We show that our approach generalizes beyond the goal of minimizing group loss to other interesting settings.
\textbf{Game 2} (Efficient but Unfair): $\quad \min_{x_1 \in \mathbb{R}} x_1^2, \quad \min_{x_2 \in \mathbb{R}} x_2^2  - \frac{11}{10} x_1^2$. \label{ex:bad_welfare} \\\\
The minimal total loss solution of Game~2 is $(x_1,x_2) = (\pm \infty, 0)$ where $x_1$ achieves infinite loss and $y$ achieves negative infinite loss. On the other hand, the Nash equilibrium is $(x_1,x_2) = (0,0)$ with a loss of zero for both agents. This hypothetical game may also arise if a loss is mis-specified. For example, $x_1$'s true loss may have been $2x_1^2$ implying no inequity issue with total loss minimization in the original game.
The inequity of the cooperative solution to Game~2 may be undesirable. D3C converges to losses of $1.079$ and $-1.162$ for $x_1$ and $x_2$ respectively (sum is $-0.083$) with $x_1$ shifting its relative loss attention to $\frac{A_{11}}{A_{12}} \approx \frac{11}{10}$ effectively halting training.

\subsection{Limits of a Local Update}
We use a 2-player bilinear matrix game to highlight the limitations of a local $\rho$-minimization approach. Consider initializing $A_{ij} = \frac{1}{2}$ so that the agents are purely cooperative. Even in this scenario, there are games where the agents minimizing local $\rho$ will get stuck in local, suboptimal minima of the total loss landscape. Consider the following game transformed into an optimization problem via $A_{ij} = \frac{1}{2}$:
\begin{align}
    &\min_{\boldsymbol{x}_1} \boldsymbol{x}_1^\top B_1 \boldsymbol{x}_2 \quad \min_{\boldsymbol{x}_2} \boldsymbol{x}_1^\top B_2 \boldsymbol{x}_2 \implies \min_{\boldsymbol{x}_1} \min_{\boldsymbol{x}_2} \boldsymbol{x}_1^\top (B_1+B_2) \boldsymbol{x}_2 = \boldsymbol{x}_1^\top C \boldsymbol{x}_2 = f_C(\boldsymbol{x}_1,\boldsymbol{x}_2)
\end{align}
with $\boldsymbol{x}_1, \boldsymbol{x}_2 \in \Delta^1$. Let $C = \begin{bmatrix} a , b ; c , d \end{bmatrix}$. Then the Hessian of the cooperative objective $f_C(\boldsymbol{x}_1,\boldsymbol{x}_2)$ has eigenvalues $\pm |a-b-c+d|$.
This function is generally a saddle with possibly two local minima. For example, set $a=d=0$, $b=-\frac{3}{4}$, and $c=-1$. With random initializations, gradient descent will converge to $(p,q)=(1,0)$ $\frac{3}{7}$ of the time with a value of $b$, else $(p,q)=(1,0)$ with a value of $c$, so we cannot expect local $\rho$-minimization to solve 2-player bilinear matrix games, in general, either.

\section{Agents}
\label{agents}

\subsection{Hyperparameters}


\begin{table}[ht]
    \centering
    \begin{tabular}{l|c|c|c|c|c|c|c|c|c}
        \toprule
        Game & $\eta_{A}$ & $\delta$ & $\nu$ & $\tau_{\min}$ & $\tau_{\max}$ & $A^0_i$ & $\epsilon$ & $l$ & $h$ \\ \hline
        Trust-Your-Brother & $1.0$ & $1.0$ & $0.0$ & $10$ & $20$ & $0.99$ & $0.0$ & $-5$ & $5$ \\
        Coins/Cleanup/HarvestPatch & $10^{-3}$ & $10^{-1}$ & $10^{-6}$ & $5$ & $10$ & $0.99$ & $100.0$ & $-5$ & $5$ \\
        \bottomrule
    \end{tabular}
    \caption{D3C hyperparameter settings for Algorithm~\ref{alg_rl_top}.}
    \label{tab:alg_rl_top_hyps}
\end{table}

\textbf{Trust-Your-Brother}: The reinforcement learning algorithm, $\mathbb{L}$, used for D3C in Trust-Your-Brother is REINFORCE~\citep{williams1992simple}. Policy gradients are computed using batches of $10$ episodes (full Monte Carlo returns, discount $\gamma=1$). Each batch of $10$ episodes contains $5$ episodes initialized with one prey closer to the predator, having only one grid space between itself and the predator. The other $5$ episodes swap the prey so that each is attacked an equivalent number of times. Both prey always start in adjacent cells. The baseline subtracted from the returns is computed from linear value function. This value function is trained via temporal difference learning with a learning rate $0.1$. The learning rate for REINFORCE is $0.1$.

\textbf{Coins/Cleanup/HarvestPatch}: The reinforcement learning algorithm, $\mathbb{L}$, used for D3C in Coins, Cleanup, and HarvestPatch (\S\ref{harvestpatch}) is A2C with V-trace~\citep{espeholt2018impala}.
\begin{table}[ht]
    \centering
    \begin{tabular}{l|c}
        \toprule
        Hyperparameter & Value \\ \hline
        Entropy regularization & $0.003$ \\
        Baseline loss scaling & $0.5$ \\
        Unroll length & $100$ \\
        Discount ($\gamma$) & $0.98$ \\
        RMSProp learning rate & $0.0004$ \\
        RMSProp epsilon ($\epsilon$) regularization parameter & $10^{-5}$ \\
        RMSProp momentum & $0.0$ \\
        RMSProp decay & $0.99$ \\
        \bottomrule
    \end{tabular}
    \caption{A2C hyperparameter settings for Coin, Cleanup, and HarvestPatch domains. No tuning or hyperparameter search was performed \textemdash these were default values used by our RL stack.}
    \label{tab:impala_hyps}
\end{table}

\section{Miscellaneous}

\subsection{Stealing vs Altruism}
In our proposed mixing scheme, each agent $i$ updates $A_i \in \Delta^{n-1}$ and transformed losses are defined as $\boldsymbol{f}^A = A^\top \boldsymbol{f}$. This can be interpreted as each agent $i$ deciding how to redistribute its losses over the other agents. In other words, if the loss is positive, agent $i$ is deciding who to steal from (give loss equals steal reward).

Alternatively, we could define a scheme where each agent $i$ updates $A_i$, however, the transformed losses are now defined as $\boldsymbol{f}^A = A \boldsymbol{f}$ and the columns of $A$ lie on the simplex. This scenario corresponds to agents taking on the losses of other agents. In other words, again assuming positive losses, deciding which agents to help. In experiments on the prisoner's dilemma, this approach did not make significant progress towards minimizing the price of anarchy so we discontinued its use in further experiments. In theory, this approach should be viable; it just requires that the information contained in agent $j$'s loss is enough to accelerate descent of agent $i$'s loss faster than the immediate loss (debt) that agent $i$ takes on.

\subsubsection{Towards A Market of Agents}
Expanding on this last perspective, when D3C agents, as defined in the main body, steal from other agents, they are exchanging immediate reward for information. The agent that is ``stolen from" receives a loss signal that can then be used to derive policy update directions. The agent that is ``stealing" receives immediate relief of loss, a form of payment. This exchange forms some of the components critical for a market economy of agents. The essential missing component is the negotiation phase where agents can choose to opt in or out of the exchange. In the current setting, the agent who steals is always able to force a transaction.

\subsection{Reciprocity in Coin Domain}
\label{coin_reciprocity}

To evaluate the extent to which there was a pattern of reciprocity in agents' relative reward attention (i.e., the attention shifted synchronously), we conduct a permutation analysis. This permutation analysis estimates the probability that the level of synchrony we observe results from random chance.

We measure the synchrony between relative reward attention trajectories through co-integration \citep{murray1994drunk}. Co-integration allows us to estimate the synchrony between two timeseries. To do so, we take the discrete differences within each timeseries and then take the correlation of those two sequences of differences. If the timeseries are correlated, their movements should be correlated. This produces a set of co-integration coefficients ranging from $0.19$ to $0.34$ (see Figure~\ref{fig:sync_test}, red).

To ensure that we are not overestimating the significance of these patterns, we employ a permutation analysis \citep{tibshirani1993introduction}. We resample the trajectories to calculate all possible values of co-integration coefficients (see Figure~\ref{fig:sync_test}, blue). Comparing the real set against the full resampled set allows us to evaluate how extreme the real values are, under the assumption that there is no relationship between the two curves. The actual co-integration coefficients are the most extreme values across the full distribution of coefficients. To estimate the overall probability of this occurring, we evaluate the harmonic mean \textit{p}-value \citep{wilson2019harmonic}. We find that the level of synchrony observed between the relative reward attention of co-learning agents significantly deviates from chance levels with $p = 0.018$.

\begin{figure}[ht!]
    \centering
    \includegraphics[scale=0.5]{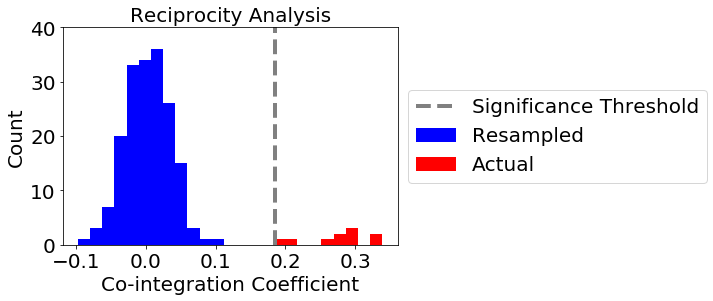}
    \caption{Histogram of co-integration coefficients for actual and resampled relative reward attention trajectories.}
    \label{fig:sync_test}
\end{figure}

\subsection{Convex Optimization vs Smooth $1$-Player Games}
\begin{proposition}A convex loss function is not necessarily a smooth game where the players are interpreted as the elements of the variable to be minimized.
\end{proposition}
\begin{proof}
    Consider the following game:
    \begin{align}
        &\min_x (x+y)^2 \quad \quad \min_y (x+y)^2.
    \end{align}
    Recall the definition of a smooth game (Definition~\ref{def:smooth_game}) and let $x=y=0$ and $x'=-y'=c$. The game is not smooth for $c > 0$ for any $\lambda, \mu$ even though this is a convex optimization problem.
\end{proof}

\subsection{Games with Mixing-Agnostic Universally-Stable Nash}
\label{incentive_compat}
Define the gradient map, $F^A$, and its Jacobian, $J^A$, for a game with loss vector $\boldsymbol{f}$ concisely with
\begin{align}
    F^A(\boldsymbol{x}) &= \begin{bmatrix}
    \langle A_i, \nabla_{x} \boldsymbol{f}(\boldsymbol{x}) \rangle
    \end{bmatrix} \\
    &= \begin{bmatrix}
    \sum_j A_{ij} \frac{\partial f_j}{\partial x_i}
    \end{bmatrix} \\
    J^A(\boldsymbol{x}) &= \begin{bmatrix}
    \sum_j A_{ij} \frac{\partial^2 f_j}{\partial x_i \partial x_k}
    \end{bmatrix} \\
    &= \begin{bmatrix}
    \sum_j A_{ij} H^j_{ik}
    \end{bmatrix}
\end{align}
where $H^j$ is the Hessian of $f_j(\boldsymbol{x})$.

\begin{proposition}
\label{diag_dom}
If each $H^j$ is diagonally dominant, then $J^A$ is diagonally dominant.
\end{proposition}
\begin{proof}
We are given $H^j_{ii} > \sum_{k \ne i} \vert H^j_{ik} \vert$. Then
\begin{align}
    J^A_{ii} = \sum_j A_{ij} H^j_{ii} &> \sum_j A_{ij} \sum_{k \ne i} \vert H^j_{ik} \vert \text{ by given \& $A_{ij}\ge 0$} \\
    &= \sum_j \sum_{k \ne i} \vert A_{ij} H^j_{ik} \vert \text{ by $A_{ij} \ge 0$} \\
    &= \sum_{k \ne i} \sum_j \vert A_{ij} H^j_{ik} \vert \text{ swap sums} \\
    &\ge \sum_{k \ne i} \vert \sum_j A_{ij} H^j_{ik} \vert \text{ by $\Delta$-inequality} \\
    &= \sum_{k \ne i} \vert J^A_{ik} \vert.
\end{align}
\end{proof}

\begin{proposition}
If each $H^j$ is diagonally dominant and $\mathcal{X}$ is unconstrained (i.e., $\mathbb{R}^d$ for some $d$), then $\boldsymbol{x}^*_A$ is the Nash equilibrium of the transformed game (i.e., with loss vector $\boldsymbol{f}$ transformed by $A$) .
\end{proposition}

\begin{proof}
Proposition~\ref{diag_dom} implies the dynamical system $\dot{\boldsymbol{x}}=-F^A(\boldsymbol{x})$ is globally stable at $\boldsymbol{x}^*_A$ for every fixed $A$. Proposition~\ref{diag_dom} also implies that each loss in the transformed game is convex. This is because $J^A_{ii}$ is the Hessian of each loss $i$ in the new game, and we showed these are positive.  Moreover, the unique fixed point of an unconstrained game with convex losses is the solution to a suitably defined variational inequality: VI($F^A, \mathbb{R}^d$). This, in turn, implies that the fixed point is the Nash equilibrium of the game~\citep{cavazzuti2002nash}.
\end{proof}




\end{document}